\documentclass[a4paper,fleqn,11pt]{article}

\usepackage{amsmath}
\usepackage{amsthm}
\usepackage{amssymb}
\usepackage{mathtools}
\usepackage{comment}
\usepackage{accents}

\usepackage[a4paper,top=3cm, bottom=3cm, left=3cm, right=3cm]{geometry}
\usepackage[shortlabels,inline]{enumitem}
\usepackage[pdftex]{graphicx}
\usepackage[pdftex,ocgcolorlinks,pagebackref=false]{hyperref}
\usepackage{authblk}

\setlist[enumerate,1]{label={(\roman*)}}

\usepackage[capitalise,nameinlink,noabbrev]{cleveref}

\Crefname{property}{Property}{Properties}

\theoremstyle{plain}
\newtheorem{theorem}{Theorem}[section]
\newtheorem{lemma}[theorem]{Lemma}
\newtheorem{proposition}[theorem]{Proposition}
\newtheorem{corollary}[theorem]{Corollary}
\newtheorem{remark}[theorem]{Remark}
\newtheorem{example}[theorem]{Example}

\theoremstyle{definition}
\newtheorem{definition}[theorem]{Definition}

\newcommand{\setbuild}[2]{\left\{#1\middle|#2\right\}}

\newcommand{\norm}[2][]{\left\|#2\right\|_{#1}}

\newcommand{\ket}[1]{\left|#1\right\rangle}
\newcommand{\bra}[1]{\left\langle #1\right|}
\newcommand{\ketbra}[2]{\left|#1\middle\rangle\!\middle\langle#2\right|}

\newcommand{\vectorstate}[1]{\ketbra{#1}{#1}}

\DeclareMathOperator{\boundeds}{\mathcal{B}}
\DeclareMathOperator{\Hom}{Hom}
\DeclareMathOperator{\Tr}{Tr}

\newcommand{\entropy}{H}
\newcommand{\relativeentropy}[3][]{\mathop{D_{#1}}\mathopen{}\left(#2\middle\|#3\right)\mathclose{}}
\newcommand{\thesandwicheddivergence}[1][]{\tilde{D}_{#1}}
\newcommand{\sandwiched}[3][]{\mathop{\thesandwicheddivergence[#1]}\mathopen{}\left(#2\middle\|#3\right)\mathclose{}}

\newcommand{\themaxrelative}{\tilde{D}_{\infty}}
\newcommand{\maxrelative}[2]{\mathop{\themaxrelative}\mathopen{}\left(#1\middle\|#2\right)\mathclose{}}

\newcommand{\distributions}[1][]{\mathcal{P}_{#1}}
\newcommand{\unittensor}[1]{\langle{#1}\rangle}

\newcommand{\partitions}[1][]{\mathcal{P}_{#1}}

\newcommand{\reals}{\mathbb{R}}
\newcommand{\complexes}{\mathbb{C}}
\newcommand{\naturals}{\mathbb{N}}
\newcommand{\nonnegativereals}{\mathbb{R}_{\ge 0}}
\newcommand{\positivereals}{\mathbb{R}_{>0}}

\newcommand{\preorderle}{\preccurlyeq}
\newcommand{\preorderge}{\succcurlyeq}

\newcommand{\asymptoticge}{\gtrsim}

\DeclareMathOperator{\supp}{supp}

\DeclareMathOperator{\rank}{rk}

\newcommand{\ubar}[1]{\underaccent{\bar}{#1}}
\renewcommand{\complement}[1]{\overline{#1}}

\newcommand{\geometricmean}{\mathbb{G}}
\newcommand{\symmetricpower}[1][]{\textnormal{Sym}^{#1}}

\DeclareMathOperator{\tensorrank}{R}
\DeclareMathOperator{\abstractrank}{R}
\DeclareMathOperator{\abstractsubrank}{Q}
\DeclareMathOperator{\abstractasymptoticrank}{\undertilde{R}}
\DeclareMathOperator{\abstractasymptoticsubrank}{\undertilde{Q}}

\title{Interpolating between R\'enyi entanglement entropies for arbitrary bipartitions via operator geometric means}

\author[1,2]{D\'avid Bug\'ar}
\author[1,2]{P\'eter Vrana}
\affil[1]{Department of Geometry, Institute of Mathematics, Budapest University of Technology and Economics, M\H uegyetem~rkp. 3., H-1111 Budapest, Hungary.}
\affil[2]{MTA-BME Lend\"ulet Quantum Information Theory Research Group, M\H uegyetem~rkp. 3., H-1111 Budapest, Hungary}

\date{\today}

\begin{document}
\maketitle

\begin{abstract}
The asymptotic restriction problem for tensors can be reduced to finding all parameters that are normalized, monotone under restrictions, additive under direct sums and multiplicative under tensor products, the simplest of which are the flattening ranks. Over the complex numbers, a refinement of this problem, originating in the theory of quantum entanglement, is to find the optimal rate of entanglement transformations as a function of the error exponent. This trade-off can also be characterized in terms of the set of normalized, additive, multiplicative functionals that are monotone in a suitable sense, which includes the restriction-monotones as well. For example, the flattening ranks generalize to the (exponentiated) R\'enyi entanglement entropies of order $\alpha\in[0,1]$. More complicated parameters of this type are known, which interpolate between the flattening ranks or R\'enyi entropies for special bipartitions, with one of the parts being a single tensor factor.

We introduce a new construction of subadditive and submultiplicative monotones in terms of a regularized R\'enyi divergence between many copies of the pure state represented by the tensor and a suitable sequence of positive operators. We give explicit families of operators that correspond to the flattening-based functionals, and show that they can be combined in a nontrivial way using weighted operator geometric means. This leads to a new characterization of the previously known additive and multiplicative monotones, and gives new submultiplicative and subadditive monotones that interpolate between the R\'enyi entropies for all bipartitions. We show that for each such monotone there exist pointwise smaller multiplicative and additive ones as well. In addition, we find lower bounds on the new functionals that are superadditive and supermultiplicative.
\end{abstract}

\section{Introduction}\label{sec:introduction}

This work is motivated by the asymptotic restriction problem of tensors, and by the study of multipartite asymptotic entanglement transformations in the strong converse domain, related to the approach of \cite{strassen1986asymptotic,strassen1987relative,strassen1988asymptotic,strassen1991degeneration} and \cite{christandl2021universal,jensen2019asymptotic,vrana2020family}. Both of these problems concern tensors of some fixed order $k$, i.e. elements of a tensor product of $k$ vector spaces $V_1\otimes\dots\otimes V_k$ over a field $\mathbb{F}$, in the second case the vector spaces being finite-dimensional complex Hilbert spaces. The tensor product of two tensors $a\in V_1\otimes\dots\otimes V_k$ and $b\in W_1\otimes\dots\otimes W_k$ is the order-$k$ tensor $a\otimes b\in(V_1\otimes W_1)\otimes\dots\otimes(V_k\otimes W_k)$. The (tensor) direct sum of two tensors $a\in V_1\otimes\dots\otimes V_k$ and $b\in W_1\otimes\dots\otimes W_k$ is the order-$k$ tensor $a\otimes b\in(V_1\oplus W_1)\otimes\dots\otimes(V_k\oplus W_k)$ (by the distributive law, this tensor product contains $(V_1\otimes\dots\otimes V_k)\oplus(W_1\otimes\dots\otimes W_k)$). The unit tensor $\unittensor{r}$ is the element $\sum_{i=1}^r e_i\otimes\dots\otimes e_i\in\mathbb{F}^r\otimes\dots\otimes\mathbb{F}^r$, where $e_1,\dots,e_r$ is the standard basis of $\mathbb{F}^r$.

We say that $a$ restricts to $b$ and write $a\ge b$ (restriction preorder) if there exist linear maps $A_j:V_j\to W_j$ such that $(A_1\otimes\dots\otimes A_k)a=b$. The tensor $a$ asymptotically restricts to $b$ if there is a sequence $(r_n)_{n\in\naturals}$ of natural numbers such that $\sqrt[n]{r_n}\to 1$ and $\unittensor{r_n}\otimes a^{\otimes n}\ge b^{\otimes n}$ for all $n$. A similar but more complicated preorder is defined in the context of the entanglement transformation problem (see \cite{jensen2019asymptotic} for details), which we will not use explicitly. For simplicity, we will write $a\asymptoticge b$ both for the asymptotic restriction and the corresponding notion for entanglement transformations.

A common trait of the two problems is that both asymptotic preorders admit a dual characterization in terms of $\nonnegativereals$-valued parameters of tensors that are $\otimes$-multiplicative, $\oplus$-additive, $\le$-monotone and map $\unittensor{r}$ to the number $r$. The set of all such parameters is the \emph{asymptotic spectrum of tensors} \cite{strassen1988asymptotic} (respectively, the \emph{asymptotic spectrum of LOCC transformations} \cite{jensen2019asymptotic}). More precisely, $a\asymptoticge b$ iff $f(a)\ge f(b)$ for every element $f$ in the asymptotic spectrum, highlighting the importance to understand these objects.

The set $[k]$ can be partitioned into an unordered pair of two nonempty subsets in $2^{k-1}-1$ ways. For each such bipartiton, we can view an order-$k$ tensor as an order-$2$ tensor (``flattening''), and consider its rank (matrix rank). These parameters provide points in the asymptotic spectrum of tensors, called the \emph{gauge points} (\emph{Eichpunkte}) \cite[Equation (3.10)]{strassen1988asymptotic}. In the case of entanglement transformations, in addition to the matrix rank, sums of powers of the singular values of each of these matrices are elements of the asymptotic spectrum of LOCC transformations, with exponents $2\alpha\in[0,2]$. In quantum information terms, these are exponentiated R\'enyi entropies of entanglement \cite{vidal2000entanglement} across each bipartition.

For $k=3$, the asymptotic spectrum of the subsemiring generated by the matrix multiplication tensors is log-star-convex with respect to the gauge points and it is conjectured to be log-convex \cite{strassen1988asymptotic}. For $k\ge 4$, Wigderson and Zuiddam \cite{wigderson2021asymptotic} extended this to tensor networks on a fixed graph. These results and the known elements of the asymptotic spectrum of (subsemirings of) tensors point to the conjecture that the gauge points are the vertices of a convexly-parametrized family of spectral points. Likewise, the exponentiated R\'enyi entropies of entanglement across all the different bipartitions may be part of a single convexly-parametrized family. In this paper, we make progress in this problem by constructing a family of subadditive and submultiplicative monotone quantities that contains every known element of the asymptotic spectrum of LOCC transformations and is closed under a suitably-defined version of convex combinations when the order parameter $\alpha$ is at least $1/2$. We also show that any subadditive and submultiplicative monotone functionals gives rise to a closely related functional with the same properties that in addition is equal to the maximum of a suitable set of spectrum elements, and these sets are disjoint for different values of the convex parameters.

\subsection{Overview of results}

The key objects in this paper are families of observables $A=(A_{\mathcal{H},n})_{\mathcal{H},n}$ where $A^\alpha_{\mathcal{H},n}\in\boundeds(\symmetricpower[n](\mathcal{H}))$, indexed by $k$-partite Hilbert spaces $\mathcal{H}=\mathcal{H}_1\otimes\dots\otimes\mathcal{H}_k$ and natural numbers $n$ subject to certain conditions relating the observables for different $\mathcal{H}$ and $n$ (see \cref{it:observableisometry,it:observablebound,it:observablesupermultiplicative,it:observabletensorproduct,it:observabledirectsum} in \cref{sec:upper}).
\begin{itemize}
\item We prove that, given a functional $g$ on a semiring with a Strassen preorder that is subadditive, submultiplicative, normalized, and monotone, there is a unique maximal functional $\tilde{g}\le g$, that in addition to these properties is multiplicative and additive on any subsemiring generated by one element, and that $\tilde{g}$ is equal to the maximum of all elements of the asymptotic spectrum that are less than $g$. In particular, the set of spectrum elements less than $g$ is nonempty (\cref{thm:reguppermaximum}).
\item We prove that, for every suitable family $A$ and $\alpha\in[0,1)$, the quantity $F^{\alpha,A}(\psi)=2^{(1-\alpha)E^{\alpha,A}(\psi)}$ where
\begin{equation}
E^{\alpha,A}(\psi)=\lim_{n\to\infty}\frac{1}{n}\frac{\alpha}{1-\alpha}\log\bra{\psi}^{\otimes n}A^{\frac{1-\alpha}{\alpha}}_{\mathcal{H},n}\ket{\psi}^{\otimes n}=-\lim_{n\to\infty}\frac{1}{n}\sandwiched[\alpha]{\vectorstate{\psi}^{\otimes n}}{A_{\mathcal{H},n}}
\end{equation}
is a subadditive and submultiplicative monotone, where $\thesandwicheddivergence[\alpha]$ is the minimal or sandwiched R\'enyi divergence \cite{wilde2014strong,muller2013quantum,tomamichel2015quantum} (\cref{thm:upper}).
\item Based on the decomposition of tensor power spaces and the Schur--Weyl duality, we construct a family of observables for $k=2$, which satisfies these conditions (\cref{prop:bipartiteobservables}). The corresponding functionals are in fact spectral points, and recover the entire asymptotic spectrum of LOCC transformations in the bipartite case, the exponentiated R\'enyi entropies of entanglement of orders $\alpha\in[0,1]$. (\cref{lem:bipartiteasymptotic})
\item We find that any finite collection of such families of observables can be combined using multivariate weighted operator geometric means to obtain new families, an operation in the spirit of various log-convexity conjectures. (\cref{prop:geometricobservables})
\item By grouping the tensor factors, the family of observables in the bipartite case generates $2^{k-1}-1$ families for general $k$. Geometric means of these give rise to subadditive and submultiplicative functionals that interpolate between the bipartite exponentiated R\'enyi relative entropies for every possible bipartition. By taking the geometric means of only a subset of these (those corresponding to bipartitions of the form $\{\{j\},[k]\setminus\{j\}\}$), we find a new expression for the previously known elements of the asymptotic spectrum of tensors and of LOCC transformations. (\cref{prop:extension})
\item When $\alpha\in[1/2,1)$, we find upper and lower bounds on the monotones that coincide when every flattening is proportional to a partial isometry (i.e. every marginal is proportional to a projection). The lower bound is a lower functional, giving us additive and multiplicative monotones on the subsemiring generated by such elements. (\cref{prop:lowerfunctionallowerbound,cor:specialsubsemiringmonotone})
\end{itemize}

The paper is structured as follows. In \cref{sec:preliminaries} we collect relevant facts on R\'enyi entropies and divergences, preordered semirings and asymptotic spectra, the representation theory of symmetric and unitary groups, and on multivariate operator geometric means. In \cref{sec:abstractfunctionals} we study monotone functionals on preordered semirings that are subadditive and submultiplicative, which puts our main findings in a general context. In \cref{sec:upper} we construct and study the subadditive and submultiplicative functionals obtained from families of observables. In \cref{sec:lower} we construct bounds on those functionals and show their consequences related to the asymptotic spectra.

\subsection{Relation to previous work}

It is known that there exist additive and multiplicative monotone functionals interpolating between special gauge points corresponding to bipartitions of the form $\{\{j\},[k]\setminus\{j\}\}$. We provide a brief review of the known functionals to provide a more complete context for the present work.
\paragraph{Support functionals.}

In \cite{strassen1991degeneration} Strassen defined the support functionals as follows. Let $C=((v_{1,1},\dots,v_{1,d_1}),\dots,(v_{k,1},\dots,v_{k,d_k}))$, where $(v_{j,1},\dots,v_{j,d_j})$ is a basis for $V_j$ for each $j\in[k]$ ($d_j=\dim V_j$). Every tensor $f\in V_1\otimes\dots\otimes V_k$ can be uniquely written as
\begin{equation}
\sum_{i_1=1}^{d_1}\dots\sum_{i_k=1}^{d_k}f_{i_1\dots i_k}v_{1,i_1}\otimes\dots\otimes v_{k,i_k},
\end{equation}
where $f_{i_1\dots i_k}\in\mathbb{F}$ for all $i_1,\dots,i_k$. The \emph{support of $f$ with respect to $C$} is
\begin{equation}
\supp_C(f)=\setbuild{(i_1,\dots,i_k)\in[d_1]\times\dots\times[d_k]}{f_{i_1\dots i_k}\neq 0}.
\end{equation}
The set of $k$-tuples of bases will be denoted by $\mathcal{C}(f)$ (where the dependence on $f$ is through the vector spaces $V_1,\dots,V_k$).

If $\theta\in\distributions([k])$ and $\Phi\subseteq I_1\times\dots\times I_k$, where $I_1,\dots,I_k$ are finite sets, then we define
\begin{equation}\label{eq:entropytheta}
\entropy_\theta(\Phi)=\max_{P\in\distributions(\Phi)}\sum_{j=1}^k\theta(j)\entropy(P_j),
\end{equation}
where $P_j$ is the $j$th marginal of $P$ and $\entropy$ is the Shannon entropy.

The \emph{maximal points} of $\Phi\subseteq I_1\times\dots\times I_k$ are defined with respect to the product partial order: $i=(i_1,\dots,i_k)\le i'=(i'_1,\dots,i'_k)$ iff $i_j\le i'_j$ for all $j\in[k]$. Explicitly,
\begin{equation}
\max\Phi=\setbuild{i\in\Phi}{\forall i'\in\Phi:i\le i'\implies i=i'}.
\end{equation}

For convex weights $\theta\in\distributions([k])$, the \emph{upper support functional} is
\begin{align}
\zeta^\theta(f) & = \begin{cases}
2^{\rho^\theta(f)} & \text{if $f\neq 0$}  \\
0 & \text{if $f=0$}
\end{cases}  \\
\rho^\theta(f) & = \min_{C\in\mathcal{C}(f)}\entropy_\theta(\supp_C(f)),
\end{align}
and the \emph{lower support functional} is
\begin{align}
\zeta_\theta(f) & = \begin{cases}
2^{\rho_\theta(f)} & \text{if $f\neq 0$}  \\
0 & \text{if $f=0$}
\end{cases}  \\
\rho_\theta(f) & = \max_{C\in\mathcal{C}(f)}\entropy_\theta(\max\supp_C(f)).
\end{align}

The upper support functional is monotone, normalized, additive, and submultiplicative \cite[Theorem 2.8]{strassen1991degeneration}, while the lower support functional is monotone, normalized, superadditive, and supermultiplicative \cite[Theorem 3.5]{strassen1991degeneration}, and $\zeta_\theta(f)\le\zeta^\theta(f)$ \cite[Corollary 4.3]{strassen1991degeneration}. On the subsemiring of \emph{oblique} (\emph{schr\"ag}) tensors $\zeta_\theta=\zeta^\theta$, therefore it is multiplicative and additive.

When the parameter $\theta$ is extremal, both support functionals reduce to one of the $k$ gauge points (flattening ranks) that we obtain if we group the tensor factors as $V_j\otimes\left(\bigotimes_{j'\neq j}V_{j'}\right)$. It does not seem to be possible to extend the definitions of the support functionals to include other gauge points. In particular, including the entropies of more general marginals $P_S$ for $S\subseteq[k]$ in \eqref{eq:entropytheta} does not lead to such an extension, because the Shannon entropy is monotone in the sense that $S\subseteq T$ implies $\entropy(P_S)\le\entropy(P_T)$, but the $S$-flattening rank may be greater than the $T$-flattening rank.

\paragraph{Quantum functionals.}
Over the field $\mathbb{C}$, the quantum functionals are points in the asymptotic spectrum of all tensors that agree with the support functionals on free tensors \cite{christandl2021universal}. Let $B_k$ denote the set of \emph{bipartitions} of $[k]$, i.e. unordered pairs $b=\{S,[k]\setminus S\}$ where $S\subseteq[k]$, $S\neq\emptyset$ and $S\neq[k]$. The bipartitions $\{S,[k]\setminus S\}$ and $\{T,[k]\setminus T\}$ are \emph{noncrossing} if $S\subseteq T$ or $T\subseteq S$ or $S\cap T=\emptyset$.

To each $n$ and partition $\lambda$ of $n$ we associate the projection
\begin{equation}
P^{V_b}_\lambda=(P^{V_S}_\lambda\otimes I_{[k]\setminus S})P^{V_{[k]}}_{(n)},
\end{equation}
where e.g. $V_S=\bigotimes_{j\in S}V_j$ and $P^W_\lambda$ is the isotypic projection on $W^{\otimes n}$ corresponding to the partition $\lambda$ (see \cref{sec:SchurWeyl} for details). If $b_1$ and $b_2$ are noncrossing bipartitions then $P^{V_{b_1}}_{\lambda_1}$ and $P^{V_{b_2}}_{\lambda_2}$ commute.

Let $\theta\in\distributions(B_k)$ such that any pair of bipartitions in its support are noncrossing. The \emph{upper quantum functional} is $F^\theta(t)=2^{E^\theta(t)}$, where
\begin{equation}
E^\theta(t)=\sup\setbuild{\sum_{b\in\supp\theta}\theta(b)\entropy(\lambda^{(b)}/n)}{n\in\naturals,\forall b\in\supp\theta:\lambda^{(b)}\vdash n,\prod_{b\in\supp\theta}P^{V_b}_{\lambda^{(b)}}t^{\otimes n}\neq 0}.
\end{equation}
Let $\theta\in\distributions(B_k)$ be arbitrary. The \emph{lower quantum functional} is $F_\theta(t)=2^{E_\theta(t)}$, where
\begin{equation}
E_\theta(t)=\sup_{\psi=(A_1\otimes\dots\otimes A_k)t}\sum_{b=\{S,[k]\setminus S\}\in\supp\theta}\theta(b)\entropy\left(\frac{\Tr_S\vectorstate{\psi}}{\norm{\psi}^2}\right),
\end{equation}
the supremum being over linear maps $A_1,\dots,A_k$ from $V_j$ to Hilbert-spaces $\mathcal{H}_j$ (which can be assumed to be of dimension $\dim V_j$).

The upper quantum functional is monotone, normalized, subadditive, and submultiplicative, while the lower quantum functional is monotone, normalized, superadditive, and supermultiplicative, and $F_\theta(f)\le F^\theta(f)$. If $
\theta$ is supported on bipartitions of the form $\{\{j\},[k]\setminus\{j\}\}$, then the upper and the lower functionals (with the same parameter $\theta$) coincide, therefore they are multiplicative and additive.

For general $\theta$ with pairwise noncrossing support, it is not known if the upper and lower quantum functionals are equal. A major obstacle to extending the definition of the upper quantum functionals to arbitrary $\theta$ is that the projections corresponding to crossing bipartitions do not commute. It is an open problem to find an extension to general $\theta$ that is subadditive and submultiplicative. In \cite{christandl2020weighted} Christandl, Lysikov and Zuiddam gave a rank-type characterization of the quantum functionals using a weighted version of slice rank, but it is unclear whether this approach can lead to a generalization to all bipartitions.

\paragraph{LOCC functionals.}
In \cite{chitambar2008tripartite} Chitambar, Duan, and Shi noted that asymptotic tensor restrictions over the complex numbers can be viewed as asymptotic entanglement transformations of pure multipartite entangled states via stochastic local operations and classical communication \cite{bennett2000exact}. The word ``stochastic'' means that the transformation protocol is required to succeed only with an arbitrarily low but nonzero probability, and we put no restriction on the probability as a function of the number of copies. As a refinement of this problem, one can consider the trade-off between the asymptotic behaviour of the probability and the transformation rate. In the special case of bipartite entanglement concentration, this relation has been determined in \cite{hayashi2002error} in various regimes of exact and approximate transformations.

We focus on the limit when the probability approaches zero exponentially (but the transformation is exact for any number of copies). The trade-off between the error exponent and the transformation rate can be characterized for any number of subsystems in terms of functionals on vectors of tensor product Hilbert spaces (thought of as unnormalized state vectors) that are normalized to $1$ on separable states, additive under the direct sum,  multiplicative under the tensor product, and monotone under trace-nonincreasing LOCC transformations \cite{jensen2019asymptotic}. The collection of these functionals is the \emph{asymptotic spectrum of LOCC transformations}. Any such functional has an exponent $\alpha$ that characterizes its scaling as $f(\sqrt{p}\psi)=p^\alpha f(\psi)$, and monotonicity is equivalent to the inequality $f(\psi)^{1/\alpha}\ge f(\Pi\psi)^{1/\alpha}+f((I-\Pi)\psi)^{1/\alpha}$ for all vectors $\psi$ and local orthogonal projections $\Pi$. In the $\alpha\to 0$ limit the condition is that $f$ is monotone under tensor restriction (in particular, the inner product plays no role for these monotones), therefore the asymptotic spectrum of tensors is a subset of the asymptotic spectrum of LOCC transformations.

The known elements of the asymptotic spectrum of LOCC transformations can be constructed as follows \cite{vrana2020family}. Given a vector $\psi\in\mathcal{H}_1\otimes\dots\otimes\mathcal{H}_k$ and decreasingly-ordered nonnegative vectors $\overline{\lambda_1},\dots,\overline{\lambda_k}$ with $\norm[1]{\overline{\lambda_j}}=1$, we define
\begin{equation}\label{eq:ratefunction}
I_\psi(\overline{\lambda_1},\dots,\overline{\lambda_k})=\lim_{\epsilon\to 0}\lim_{n\to\infty}-\frac{1}{n}\log\sum_{\substack{\lambda_1,\dots,\lambda_k\vdash n  \\  \forall j:\norm[1]{\frac{\lambda_j}{n}-\overline{\lambda_j}}\le\epsilon}}\norm{(P^{\mathcal{H}_1}_{\lambda_1}\otimes\dots\otimes P^{\mathcal{H}_k}_{\lambda_k})\psi^{\otimes n}}^2.
\end{equation}
Let $\alpha\in(0,1)$ and $\theta\in\distributions([k])$. The upper functional is defined for $\psi\neq 0$ as 
\begin{align}
F^{\alpha,\theta}(\psi) & =2^{(1-\alpha)E^{\alpha,\theta}(\psi)}  \label{eq:upperLOCC}
\intertext{where}
E^{\alpha,\theta}(\psi) & =\sup_{\overline{\lambda_1},\dots,\overline{\lambda_k}}\left[\sum_{j=1}^k\theta(j)\entropy(\overline{\lambda_j})-\frac{\alpha}{1-\alpha}I_\psi(\overline{\lambda_1},\dots,\overline{\lambda_k})\right].  \label{eq:logupperLOCC}
\end{align}
The lower functional in \cite{vrana2020family} is defined for $\psi\neq 0$ as 
\begin{align}
F_{\alpha,\theta}(\psi) & = 2^{(1-\alpha)E_{\alpha,\theta}(\psi)},  \label{eq:lowerLOCC}  \\
E_{\alpha,\theta}(\psi) & = \sup_{\varphi=(A_1\otimes\dots\otimes A_k)\psi}\sum_{j=1}^k\theta(j)\entropy\left(\frac{\Tr_j\vectorstate{\varphi}}{\norm{\varphi}^2}\right)+\frac{\alpha}{1-\alpha}\log\norm{\varphi}^2,
\end{align}
where the supremum is over contractions $A_j:\mathcal{H}_j\to\mathcal{H}_j$.

$F^{\alpha,\theta}$ is equal to the regularization of $F_{\alpha,\theta}$, and is an element of the asymptotic spectrum of LOCC transformations that scales as $F^{\alpha,\theta}(\sqrt{p}\psi)=p^\alpha F^{\alpha,\theta}(\psi)$. For constant $\theta$ and varying $\alpha\in(0,1)$, the functionals can be seen as a family of multipartite R\'enyi entanglement measures (monotones) interpolating between the quantum functionals ($\alpha\to 0$) and the $\theta$-weighted arithmetic mean of von~Neumann entanglement entropies for the bipartitions $\{\{j\},[k]\setminus\{j\}\}$.

It is not known whether the definition can be extended to distributions $\theta$ on the set $B_k$ of all bipartitions. As in the case of the quantum functionals, the main difficulty is that the isotypic projections corresponding to crossing bipartitions do not commute with each other, therefore \eqref{eq:ratefunction} does not generalize in any obvious way.

\section{Preliminaries}\label{sec:preliminaries}

\subsection{R\'enyi entropies and divergences}

For positive functions $P,Q:\mathcal{X}\to\nonnegativereals$ on some finite set $\mathcal{X}$ and $\alpha\in(0,1)\cup(1,\infty)$, the \emph{R\'enyi divergence} of order $\alpha$ is \cite{renyi1961measures}
\begin{equation}
\relativeentropy[\alpha]{P}{Q}=\frac{1}{\alpha-1}\log\sum_{x\in\mathcal{X}}P(x)^\alpha Q(x)^{1-\alpha}.
\end{equation}
We note that in some applications it is convenient to include an additional term proportional to the logarithm of the norm $\norm[1]{P}=\sum_{x\in\mathcal{X}}P(x)$. We do not follow this convention and use instead this form even if $P$ is not a probability distribution. If $\alpha>1$ and $\supp P\not\subseteq\supp Q$, then $\relativeentropy[\alpha]{P}{Q}=\infty$. When $\alpha\in(0,1)$, the sum is always finite but may be $0$; in this case the logarithm is defined to be $-\infty$ and divergence is $\infty$.

The \emph{R\'enyi entropy} of order $\alpha$ is related to the R\'enyi divergence as
\begin{equation}\label{eq:entropyfromdivergence}
\entropy_\alpha(P)
 = \frac{1}{1-\alpha}\log\sum_{x\in\mathcal{X}}P(x)^\alpha
 = -\relativeentropy[\alpha]{P}{I},
\end{equation}
where $I$ denotes the constant $1$ function.

When $\norm[1]{P}=1$, the $\alpha\to 1$ limits exist and are equal to the \emph{relative entropy} (or Kullback--Leibler divergence)
\begin{align}
\relativeentropy{P}{Q} & = \sum_{x\in\mathcal{X}}P(x)\log\frac{P(x)}{Q(x)}
\intertext{and the \emph{Shannon entropy}}
\entropy(P) & = -\sum_{x\in\mathcal{X}}P(x)\log P(x),
\end{align}
respectively, while the $\alpha\to\infty$ limits are the \emph{max-divergence}
\begin{align}
\maxrelative{P}{Q} & = \log\max_{x\in\mathcal{X}}\frac{P(x)}{Q(x)}
\intertext{and the \emph{min-entropy}}
\entropy_\infty(P) & = -\log\max_{x\in\mathcal{X}}P(x).
\end{align}

The R\'enyi entropies admit a variational characterization in terms of the relative entropy and the Shannon entropy \cite{arikan1996inequality,merhav1999shannon,shayevitz2011renyi}. For $\alpha\in(0,1)\cup(1,\infty]$ and $P\in\distributions(\mathcal{X})$ we have
\begin{equation}\label{eq:variationalRenyientropy}
\entropy_\alpha(P)=\frac{\alpha}{\alpha-1}\min_{Q\in\distributions(\mathcal{X})}\left[\frac{\alpha-1}{\alpha}\entropy(Q)+\relativeentropy{Q}{P}\right].
\end{equation}

We will make use of generalizations of the R\'enyi divergences to pairs of positive operators. A function $P:\mathcal{X}\to\nonnegativereals$ can be identified with a positive semidefinite operator on $\complexes^{\mathcal{X}}$ that is diagonal in the standard basis, with diagonal entries $P(x)$, i.e. $\sum_{x\in\mathcal{X}}P(x)\ketbra{x}{x}$. A quantum R\'enyi divergence of order $\alpha$ is a functional on pairs of positive operators that is invariant under isometries and reduces to $\relativeentropy[\alpha]{P}{Q}$ when its arguments are the operators corresponding to $P$ and $Q$.

The axioms in \cite{renyi1961measures} that characterize R\'enyi divergences can be adapted to quantum generalizations as well, but do not single out a unique extension \cite[Chapter 4]{tomamichel2015quantum}. There are several families of quantum R\'enyi divergences that are known to have useful properties (at least in some parameter range) or operational significance, such as the Petz \cite{petz1986quasi}, sandwiched \cite{wilde2014strong,muller2013quantum}, maximal \cite{matsumoto2018new} and $\alpha$-$z$ divergences \cite{jaksic2011entropic,audenaert2015alpha}. However, the limit $\alpha\to\infty$ is special: there is only one extension of the R\'enyi divergence of order $\infty$ that satisfies the data processing inequality, the \emph{max-divergence}
\begin{equation}
\begin{split}
\maxrelative{\rho}{\sigma}&=\begin{cases}
\log\norm[\infty]{\sigma^{-1/2}\rho\sigma^{-1/2}} & \text{if $\rho^0\le\sigma^0$}  \\
\infty & \text{otherwise.}
\end{cases}\\
 &= \inf\setbuild{\lambda\in\reals}{\rho\le 2^\lambda \sigma},
 \end{split}
\end{equation}

Equation \eqref{eq:entropyfromdivergence} can be unambiguously extended to positive operators instead of $P$ as
\begin{equation}
\entropy_\alpha(\rho)
 = \frac{1}{1-\alpha}\log\Tr\rho^\alpha
 = -\relativeentropy[\alpha]{\rho}{I},
\end{equation}
while the $\alpha\to\infty$ limit is the min-entropy
\begin{equation}
\entropy_\infty(\rho)=-\log\norm[\infty]{\rho}.
\end{equation}

There are several quantum extensions of classical R\'enyi divergences, but for our study we mention the \emph{sandwiched R\'enyi divergence} \cite{muller2013quantum,wilde2014strong}
\begin{equation}\label{eq:sandwicheddef}
    \sandwiched[\alpha]{\rho}{\sigma}=\frac{1}{\alpha-1}\log\text{Tr}\big(\rho^{\frac{1}{2}}\sigma^{\frac{1-\alpha}{\alpha}}\rho^{\frac{1}{2}} \big)^\alpha.
\end{equation}
In our application the first of the two arguments will always be a rank-$1$ operator, in which case the expression simplifies as
\begin{equation}\label{eq:reformulatedivergencetfamily}
    \sandwiched[\alpha]{\vectorstate{\psi}}{\sigma}= \frac{\alpha}{\alpha-1}\log\bra{\psi}\sigma^{\frac{1-\alpha}{\alpha}}\ket{\psi}.
\end{equation}

\subsection{Preordered semirings and the asymptotic spectrum of tensors}

The study of asymptotic preorders on preordered semirings originates in the work of Strassen on the asymptotic restriction problem \cite{strassen1988asymptotic} (see also the recent expositions in \cite{zuiddam2018algebraic,wigderson2021asymptotic}), which in turn provides a framework for a range of techniques developed in the study of the computational complexity of matrix multiplication.

A \emph{preordered semiring} is a set $S$ equipped with commutative and associative binary operations $+$ and $\cdot$ and a preorder $\preorderle$ (reflexive and transitive relation) such that $\cdot$ distributes over $+$, there exist neutral elements $0$ and $1$ for $+$ and $\cdot$, and the preorder is compatible with the operations in the sense that if $x,y,z\in S$ are such that $x\preorderge y$, then $x+z\preorderge y+z$ and $xz\preorderge yz$ hold as well. Every natural number $n$ determines an element $1+1+\dots+1\in S$ (with $n$ terms), which gives a semiring homomorphism $\naturals\to S$. We say that $\preorderle$ is a \emph{Strassen preorder} if this homomorphism is an order-embedding and, identifying its image in $S$ with $\naturals$, for every $x\in S\setminus\{0\}$ there is an $r\in\naturals$ such that $x\preorderle r$ and $rx\preorderge 1$ \cite{strassen1988asymptotic,zuiddam2018algebraic}.

The \emph{asymptotic preorder} is defined as $x\asymptoticge y$ if there exists a sequence of natural numbers $(r_n)_{n\in\naturals}$ such that $\lim_{n\to\infty}\sqrt[n]{r_n}=1$ and $r_nx^n\preorderge y^n$ for all $n$. The \emph{asymptotic spectrum} $\Delta(S,\preorderle)$ is the set of monotone semiring homomorphisms $S\to\nonnegativereals$ (\emph{spectral points}). The asymptotic preorder can be characterized in terms of the asymptotic spectrum as $x\asymptoticge y\iff\forall f\in\Delta(S,\preorderle):f(x)\ge f(y)$ \cite{strassen1988asymptotic}.

For example, equivalence classes of order-$k$ tensors over a fixed ground field $\mathbb{F}$ form a preordered semiring with operations induced by the (tensor) direct sum and the tensor product, and preorder given either by restriction or degeneration \cite{strassen1988asymptotic}. Both are Strassen preorders, and the asymptotic spectrum provides a characterization of the asymptotic restriction preorder, which includes the asymptotic rank and asymptotic subrank as special cases. Other examples include the semiring of isomorphism classes of nilpotent representations of a given finite quiver \cite{strassen2000asymptotic} and the semiring of finite simple undirected graphs \cite{zuiddam2019asymptotic}. The first of these can be used to characterize asymptotic degeneration of nilpotent representations, while the latter leads to a characterization of the Shannon capacity of graphs.

We describe the semiring of tensors in more detail. Let $k\in\naturals$, $k\ge 2$ and fix a ground field $\mathbb{F}$. By an order-$k$ tensor over $\mathbb{F}$ we mean an element of a vector space of the form $V_1\otimes\dots\otimes V_k$ where $V_j$ are finite-dimensional vector spaces over $\mathbb{F}$. The \emph{direct sum} and \emph{tensor product} of two tensors $a\in V_1\otimes\dots\otimes V_k$ and $b\in W_1\otimes\dots\otimes W_k$ are the elements $a\oplus b\in V_1\otimes\dots\otimes V_k\oplus W_1\otimes\dots\otimes W_k\subseteq(V_1\oplus W_1)\otimes\dots\otimes(V_k\oplus W_k)$ and $a\otimes b\in(V_1\otimes W_1)\otimes\dots\otimes(V_k\otimes W_k)$, considered as order-$k$ tensors as indicated. We say that $b$ is a \emph{restriction} of $a$ and write $a\ge b$ if there exist linear maps $A_j:V_j\to W_j$ such that $(A_1\otimes\dots\otimes A_k)a=b$. $a$ and $b$ are \emph{isomorphic} if $(A_1\otimes\dots\otimes A_k)a=b$ for some invertible linear maps $A_1,\dots,A_k$, and \emph{equivalent} if there exist zero tensors $a_0,b_0$ (in suitable spaces) such that $a\oplus a_0$ and $b\oplus b_0$ are isomorphic. The set of equivalence classes of tensors with operations induced by $\oplus$ and $\otimes$ and the partial order induced by $\le$ is a preordered semiring. Its asymptotic spectrum is the \emph{asymptotic spectrum of tensors}. Concretely, it is the set of parameters of tensors that are normalized on the unit tensors, additive under the direct sum, multiplicative under the tensor product, and monotone under restrictions (equivalently: degenerations).

In the refinement relevant for the study of entanglement transformations, the tensors are elements of tensor products of finite-dimensional complex Hilbert spaces, the operations are the same, while the preorder encodes the possibility of transformations via \emph{local operations and classical communication} \cite{jensen2019asymptotic}. The resulting \emph{asymptotic spectrum of LOCC transformations} is the set of functionals $f$ that are normalized on unit tensors, additive under the (orthogonal) direct sum, multiplicative under the tensor product, and such that for some $\alpha\in[0,1]$ they satisfy $f(\sqrt{p}\psi)=p^\alpha f(\psi)$ and $f(\psi)^{1/\alpha}\ge f(\Pi\psi)^{1/\alpha}+f((I-\Pi)\psi)^{1/\alpha}$ for all $p\ge 0$ and local orthogonal projection $\Pi\in\boundeds(\mathcal{H}_j)\subseteq\boundeds(\mathcal{H}_1\otimes\dots\otimes\mathcal{H}_k)$ (the inclusion is via taking the tensor product with the identity).

\subsection{Schur--Weyl decompositions}\label{sec:SchurWeyl}

Irreducible representations of the symmetric groups and irreducible polynomial representations of the general linear or unitary groups are labelled by integer partitions.
A \emph{partition} of $n\in\naturals$ is a nonincreasing sequence of nonnegative integers $\lambda=(\lambda_1,\dots,\lambda_d)$ such that $\lambda_1+\dots+\lambda_d=n$. We implicitly add or remove trailing zeroes whenever convenient, in particular to obtain vectors with the same number of components. We write $\lambda\vdash n$ if $\lambda$ is a partition of $n$. The number of nonzero entries in $\lambda$ is its \emph{length} $l(\lambda)$.

We denote the irreducible representation of $S_n$ corresponding to the partition $\lambda$ by $[\lambda]$, and the irreducible representation of $U(\mathcal{H})$ corresponding to $\lambda$ by $\mathbb{S}_\lambda(\mathcal{H})$ (if $l(\lambda)>\dim\mathcal{H}$ then we set $\mathbb{S}_\lambda(\mathcal{H})=0$). The tensor powers of the standard representation of $U(\mathcal{H})$ decompose as
\begin{equation}
\mathcal{H}^{\otimes n}\simeq\bigoplus_{\lambda\vdash n}\mathbb{S}_\lambda(\mathcal{H})\otimes[\lambda]
\end{equation}
into a direct sum of irreducible representations of $U(\mathcal{H})\times S_n$ (where the symmetric group permutes the tensor factors). The orthogonal projection onto the direct summand labelled by $\lambda$ will be denoted by $P^{\mathcal{H}}_\lambda$. These are also the isotypic projections of both the $S_n$-representation and the $U(\mathcal{H})$-representation (Schur--Weyl duality). A special case is $P^{\mathcal{H}}_{(n)}$, the projection onto the symmetric subspace $\symmetricpower[n](\mathcal{H})$.

Let $0\le\rho\in\boundeds(\mathcal{H})$, and let the eigenvalues of $\rho$ be $r=(r_1,\ldots,r_d)$ with multiplicities and in decreasing order ($d=\dim\mathcal{H}$). The quantitites $\Tr\rho^{\otimes n}P^{\mathcal{H}}_\lambda$ satisfy the estimates \cite{harrow2005applications}
\begin{equation}\label{eq:quantumtypeestimate}
(n+d)^{-d(d+1)/2}2^{-n\relativeentropy{\frac{\lambda}{n}}{r}}\le\Tr\rho^{\otimes n}P^{\mathcal{H}}_\lambda\le(n+1)^{d(d-1)/2}2^{-n\relativeentropy{\frac{\lambda}{n}}{r}}.
\end{equation}

The \emph{Kronecker coefficients} are defined as
\begin{equation}
g_{\lambda\mu\nu}=\dim\Hom([\lambda],[\mu]\otimes[\nu]),
\end{equation}
and are symmetric in the three arguments. The special cases $g_{(n)\mu\nu}=0$ if $\mu\neq\nu$ and $g_{(n)\mu\mu}=1$ are consequences of Schur's lemma. The Kronecker coefficients satisfy the entropic vanishing condition \cite{christandl2006spectra}
\begin{equation}
g_{\lambda\mu\nu}\neq 0\implies \entropy(\lambda/n)\le\entropy(\mu/n)+\entropy(\nu/n).
\end{equation}

When $\mu\vdash m$, $\nu\vdash n$ and $\lambda\vdash m+n$, the \emph{Littlewood--Richardson coefficient} $c^\lambda_{\mu\nu}$ gives the multiplicity of the irreducible representation $[\mu]\otimes[\nu]$ of the group $S_m\times S_n$ in the restriction of $[\lambda]$ from $S_{m+n}$ to the Young subgroup $S_m\times S_n$. These are also subject to an entropic vanishing condition \cite{christandl2021universal}:
\begin{multline}
c^\lambda_{\mu\nu}\neq 0\implies \frac{m}{m+n}\entropy(\mu/m)+\frac{n}{m+n}\entropy(\nu/n)\le\entropy(\lambda/(m+n))  \\  \le\frac{m}{m+n}\entropy(\mu/m)+\frac{n}{m+n}\entropy(\nu/n)+h\left(\frac{m}{m+n}\right),
\end{multline}
where $h(q)=-q\log q-(1-q)\log(1-q)$ is the binary entropy.

On $\mathcal{H}^{\otimes m}\otimes\mathcal{H}^{\otimes n}\simeq\mathcal{H}^{\otimes(m+n)}$ we may consider the representations of $S_m$ and $S_n$ (left hand side) and of $S_{m+n}$ (right hand side). For all $\mu\vdash m$, $\nu\vdash n$ and $\lambda\vdash m+n$, the isotypic projections $P^{\mathcal{H}}_\mu\otimes I^{\otimes n}$, $I^{\otimes m}\otimes P^{\mathcal{H}}_\nu$ and $P^{\mathcal{H}}_\lambda$ commute with each other and satisfy
\begin{equation}
\left(P^{\mathcal{H}}_\mu\otimes P^{\mathcal{H}}_\nu\right)P^{\mathcal{H}}_\lambda\neq 0\implies c^{\lambda}_{\mu\nu}\neq 0.
\end{equation}

If $\mathcal{H}$ and $\mathcal{K}$ are two Hilbert spaces then $\mathcal{H}^{\otimes n}\otimes\mathcal{K}^{\otimes n}\simeq(\mathcal{H}\otimes\mathcal{K})^{\otimes n}$ may be decomposed with respect to $U(\mathcal{H})$, $U(\mathcal{K})$ and $U(\mathcal{H}\otimes\mathcal{K})$. For $\lambda,\mu,\nu\vdash n$ the isotypic projections $P^{\mathcal{H}\otimes\mathcal{K}}_\lambda$, $P^{\mathcal{H}}_\mu\otimes I_{\mathcal{K}}^{\otimes n}$ and $I_{\mathcal{H}}^{\otimes n}\otimes P^{\mathcal{K}}_\nu$ commute with each other and satisfy
\begin{equation}
\left(P^{\mathcal{H}}_\mu\otimes P^{\mathcal{K}}_\nu\right)P^{\mathcal{H}\otimes\mathcal{K}}_\lambda\neq 0\implies g_{\lambda\mu\nu}\neq 0.
\end{equation}
In particular,
\begin{equation}
\left(P^{\mathcal{H}}_\lambda\otimes P^{\mathcal{K}}_\lambda\right)P^{\mathcal{H}\otimes\mathcal{K}}_{(n)}=\left(P^{\mathcal{H}}_\lambda\otimes I_{\mathcal{K}}^{\otimes n}\right)P^{\mathcal{H}\otimes\mathcal{K}}_{(n)}=\left(I_{\mathcal{H}}^{\otimes n}\otimes P^{\mathcal{H}}_\lambda\right)P^{\mathcal{H}\otimes\mathcal{K}}_{(n)}.
\end{equation}

For $m,n\in\naturals$ we consider the subspace $\complexes^{\binom{m}{m+n}}\otimes\mathcal{H}^{\otimes m}\otimes\mathcal{K}^{\otimes n}\subseteq(\mathcal{H}\oplus\mathcal{K})^{\otimes (n+n)}$, and compare the decompositions with respect to $U(\mathcal{H})$, $U(\mathcal{K})$ and $U(\mathcal{H}\oplus\mathcal{K})$. For all $\mu\vdash m$, $\nu\vdash n$ and $\lambda\vdash m+n$, the isotypic projections $I_{\binom{m}{n}}\otimes P^{\mathcal{H}}_\mu\otimes I_{\mathcal{K}}^{\otimes n}$, $I_{\binom{m}{n}}\otimes I_{\mathcal{H}}^{\otimes m}\otimes P^{\mathcal{K}}_\nu$ and $P^{\mathcal{H}\oplus\mathcal{K}}_\lambda$ commute with each other and satisfy
\begin{equation}
\left(I_{\binom{m}{n}}\otimes P^{\mathcal{H}}_\mu\otimes P^{\mathcal{K}}_\nu\right)P^{\mathcal{H}\oplus\mathcal{K}}_\lambda\neq 0\implies c^{\lambda}_{\mu\nu}\neq 0.
\end{equation}

\subsection{Operator geometric means}

While there is a distinguished notion of the weighted geometric mean of a pair of positive matrices, there is no unique such choice for multivariate geometric means, even in the unweighted case. Several constructions have been proposed as well as an axiomatization by Ando--Li--Mathias \cite{ando2004geometric}. We will make use of geometric means that satisfy the properties in the following definition, which were previously considered in \cite{bunth2021equivariant}. These properties are different from the axioms in \cite{ando2004geometric}, but the known constructions based on two-variable geometric means satisfy these as well.
\begin{definition}\label{def:geometricMean}
We say that the $r$-variable operator function $\geometricmean:\boundeds(\mathcal{H})_+^r\to\boundeds(\mathcal{H})_+$ (given for all finite-dimensional Hilbert spaces $\mathcal{H}$) is a \emph{geometric mean} if it satisfies the following conditions:
\begin{enumerate}[({G}1)]
    \item\label[property]{it:gmeanunitary} $\geometricmean(UA_1U^*,\dots,UA_rU^*)=U\geometricmean(A_1,\dots,A_r)U^*$, where $U:\mathcal{H}\to\mathcal{K}$ is any unitary operator;
    \item\label[property]{it:gmeanmonotone} if $A_i\le B_i$ for all $i$, then $\geometricmean(A_1,\dots,A_r)\le\geometricmean(B_1,\dots,B_r)$;
    \item\label[property]{it:gmeantensorproduct} $\geometricmean(A_1\otimes B_1,\dots,A_r\otimes B_r)=\geometricmean(A_1,\dots,A_r)\otimes\geometricmean(B_1,\dots,B_r)$;
    \item\label[property]{it:gmeandirectsum} $\geometricmean(A_1\oplus B_1,\dots,A_r\oplus B_r)=\geometricmean(A_1,\dots,A_r)\oplus\geometricmean(B_1,\dots,B_r)$;
    \item\label[property]{it:gmeanhomogeneous} $\geometricmean(\lambda A_1,\dots,\lambda A_r)=\lambda\geometricmean(A_1,\dots,A_r)$ for all $\lambda>0$;
    \item\label[property]{it:gmeanconcave} $\geometricmean(tA_1+(1-t)B_1,\dots,tA_r+(1-t)B_r)\le t\geometricmean(A_1,\dots,A_k) +(1-t)\geometricmean(B_1,\dots,B_k)$ for all $t\in[0,1]$.
\end{enumerate}
\end{definition}

\begin{remark}\label{rem:gmeanproperties}\leavevmode
\begin{enumerate}
\item A geometric mean $\geometricmean$ determines uniquely weights $\theta\in\distributions([r])$ such that when the arguments $A_1,\dots,A_r$ commute, the geometric mean evaluates to $\geometricmean(A_1,\dots,A_r)=\prod_{i=1}^r A_i^{\theta(i)}$. Note that this applies in particular to numbers (identified with operators on $\complexes$), and also implies that $\geometricmean(I,\dots,I)=I$.
\item The Kubo--Ando means
\begin{equation}
\geometricmean(A,B)=A\#_\theta B=B^{1/2}(B^{-1/2}AB^{-1/2})^\theta B^{1/2}=A^{1/2}(A^{-1/2}BA^{-1/2})^{1-\theta}A^{1/2}
\end{equation} are bivariate geometric means \cite{pusz1975functional,kubo1980means}. More generally, nesting bivariate geometric means in an arbitrary way results in a multivariate geometric mean \cite{ando2004geometric,bini2010effective}. This allows the construction of (in general, many different) $r$-variate weighted geometric means with any weight $\theta\in\distributions([r])$.
\item By \cref{it:gmeandirectsum}, if $X$ is an operator that commutes with the arguments $A_1,\dots,A_r$ then it also commutes with $\geometricmean(A_1,\dots,A_r)$.
\item The inequality $T(\geometricmean(A_1,\dots,A_r))\le\geometricmean(T(A_1),\dots,T(A_r))$ holds for any completely positive map $T:\boundeds(\mathcal{H})\to\boundeds(\mathcal{K})$. In particular,
\begin{equation}
\bra{\psi}\mathbb{G}(A_1,\dots,A_r)\ket{\psi}\leq\mathbb{G}(\bra{\psi}A_1\ket{\psi},\dots,\bra{\psi}A_r\ket{\psi})
\end{equation}
for any vector $\psi$.
\end{enumerate}
\end{remark}

We will make use of the following lower bound for the geometric mean of positive semidefinite operators, where $\psi$ is an arbitrary vector:
\begin{equation}\label{eq:geomlowerbyvector}
\begin{split}
\geometricmean(A_1,\ldots,A_r)
 &= 2^{-\sum_{i=1}^r\theta(i)\maxrelative{\vectorstate{\psi}}{A_i}}\geometricmean(2^{\maxrelative{\vectorstate{\psi}}{A_1}}A_1,\ldots,2^{\maxrelative{\vectorstate{\psi}}{A_r}}A_r)
  \\
  &\ge 2^{-\sum_{i=1}^r\theta(i)\maxrelative{\vectorstate{\psi}}{A_i}}\geometricmean(\vectorstate{\psi},\ldots,\vectorstate{\psi})  \\
  &= 2^{-\sum_{i=1}^r\theta(i)\maxrelative{\vectorstate{\psi}}{A_i}}\vectorstate{\psi}.
\end{split}
\end{equation}

\section{Abstract upper and lower functionals}\label{sec:abstractfunctionals}

Throughout this section $(S,\preorderle)$ is a semiring with a Strassen preorder. Recall that the asymptotic spectrum $\Delta(S,\preorderle)$ characterizes the asymptotic preorder: if $h$ is a spectral point and $x\asymptoticge y$, then $h(x)\ge h(y)$, and conversely, if the inequality $h(x)\ge h(y)$ holds for all $h\in\Delta(S,\preorderle)$, then $x\asymptoticge y$. However, finding elements of $\Delta(S,\preorderle)$ tends to be a challenging problem.

In this section we study monotone functionals on preordered semirings that are either subadditive and submultiplicative or superadditive and supermultiplicative. From a practical point of view, if $f$ is superadditive and supermultiplicative and $g$ is subadditive and submultiplicative, and the two satisfy $f\le g$ (pointwise), then the asymptotic inequality $x\asymptoticge y$ still implies $g(x)\ge f(y)$, therefore such a pair of functionals provides an obstruction in a similar way as a spectral point does. On the theoretical side, constructing such functionals is often an important step in finding spectral points. In fact, as we will see below, subadditive and submultiplicative monotone functionals already imply the existence of (typically many) pointwise smaller spectral points as well, thereby providing information about $\Delta(S,\preorderle)$.
\begin{definition}\label{def:abstractlowerupper}
Let $S$ be a preordered semiring.
A map $g:S\to\nonnegativereals$ is an \emph{(abstract) upper functional} if $g(0)=0$, $g(1)=1$, and for $x,y\in S$ it satisfies $g(x+y)\le g(x)+g(y)$, $g(xy)\le g(x)g(y)$, and $x\preorderle y$ implies $g(x)\le g(y)$.

Similarly, a map $f:S\to\nonnegativereals$ is an \emph{(abstract) lower functional} if $f(0)=0$, $f(1)=1$, and for $x,y\in S$ it satisfies $f(x+y)\ge f(x)+f(y)$, $f(xy)\ge f(x)f(y)$, and $x\preorderle y$ implies $f(x)\le f(y)$.
\end{definition}
Clearly, a map $S\to\nonnegativereals$ is a spectral point iff it is both an upper and a lower functional. It follows from the definition that if $n\in\naturals$ then $f(n)\ge n$ for any lower functional $f$ and $g(n)\le n$ for any upper functional $g$.

\begin{example}\label{ex:functionalexamples}\leavevmode
\begin{enumerate}
\item In any semiring $S$ with a Strassen preorder $\preorderle$ we can define the \emph{abstract rank} as \cite[Section 3.]{zuiddam2019asymptotic}
\begin{equation}
\abstractrank(x)=\min\setbuild{n\in\naturals}{x\preorderle n}
\end{equation}
and the \emph{abstract subrank} as
\begin{equation}
\abstractsubrank(x)=\max\setbuild{n\in\naturals}{n\preorderle x}.
\end{equation}
$\abstractrank$ is an upper functional and $\abstractsubrank$ is a lower functional.
\item The \emph{asymptotic rank} $\abstractasymptoticrank(x)=\lim_{n\to\infty}\sqrt[n]{\abstractrank(x)}$ is an upper functional and the \emph{asymptotic subrank} $\abstractasymptoticsubrank(x)=\lim_{n\to\infty}\sqrt[n]{\abstractsubrank(x)}$ is a lower functional \cite{strassen1988asymptotic}.
\item Similarly, the \emph{fractional rank} is an upper functional and the \emph{fractional subrank} is a lower functional \cite[Section 3.3]{wigderson2021asymptotic}.
\item On the semiring of tensor classes, the \emph{upper support functionals} $\zeta^\theta$ are upper functionals and the \emph{lower support functionals} $\zeta_\theta$ are lower functionals \cite[Sections 2. and 3.]{strassen1991degeneration}.
\item On the semiring of graphs \cite{zuiddam2019asymptotic}, the \emph{independence number} is a lower functional (equal to the abstract subrank), while the \emph{clique cover number}, the \emph{Haemers bound} and the \emph{complement of the orthogonal rank} are upper functionals (of these the clique cover number is equal to the abstract rank).
\item On the semiring of complex tensors with the restriction preorder, the \emph{lower quantum functionals} for arbitrary convex weights are lower functionals. The \emph{upper quantum functionals} for weight supported on noncrossing bipartitions are upper functionals \cite{christandl2021universal}.
\item On the semiring of complex tensors with the LOCC preorder, the functionals defined in \eqref{eq:lowerLOCC} are lower functionals \cite{vrana2020family}.
\end{enumerate}
\end{example}

As emphasized in \cite{strassen1988asymptotic}, the asymptotic rank $\abstractasymptoticrank$ of tensors is subadditive and submultiplicative and becomes additive and multiplicative when restricted to a subsemiring generated by (i.e. the set of polynomials with natural coefficients in) a single tensor, in this respect behaving like the maximum functional on a semiring of nonnegative continuous functions on a compact space. Likewise, the asymptotic subrank $\abstractasymptoticsubrank$ is a lower functional which, when restricted to a subsemiring generated by a single tensor, becomes multiplicative and additive, so it behaves like the minimum functional. The theory of asymptotic spectra provides an explanation for this analogy in the form of the equalities $\abstractasymptoticrank(x)=\max_{h\in\Delta(S,\preorderle)}h(x)$ and $\abstractasymptoticsubrank(x)=\min_{h\in\Delta(S,\preorderle)}h(x)$ \cite[Equation (1.13)]{strassen1988asymptotic}.

The idea behind the present section is to consider instead of $\Delta(S,\preorderle)$ a (closed) subset $H\subset\Delta(S,\preorderle)$, which gives rise to an upper functional $x\mapsto\max_{h\in H}h(x)$ (an example of a functional of this form is an asymptotic version of the upper support functional \cite[Section 7, (ii)]{strassen1991degeneration}). Our aim is to characterize the upper functionals arising in this way and to relate the properties of the functional to $H$. Viewing the construction backwards, by finding a suitable upper functional, one can prove the existence of spectral points subject to certain constraints. In particular, we will show that if $\varphi$ is a spectral point on a subsemiring $S_0$ and $g$ is an upper functional such that $g|_{S_0}=\varphi$ then there exists a spectral point $h$ such that $h|_{S_0}=\varphi$ and $h\le g$ (without the upper bound, an extension of $\varphi$ was known to exist \cite[Corollary 2.18.]{zuiddam2018algebraic}).

Unlike spectral points, the sets of upper an lower functionals are closed under several basic constructions:
\begin{proposition}\label{prop:functionalbasicconstructions}\leavevmode
\begin{enumerate}
\item Let $G$ be a nonempty set of upper functionals. Then $\bar{g}(x)=\sup_{g\in G}g(x)$ is an upper functional.
\item Let $(I,\le)$ be a directed set and $(g_i)_{i\in I}$ a monotone decreasing net of upper functionals. Then $\ubar{g}(x)=\inf_{i\in I}g_i(x)$ is an upper functional.

\item Let $F$ be a nonempty set of lower functionals. Then $\ubar{f}(x)=\inf_{f\in F}f(x)$ is a lower functional.
\item Let $(I,\le)$ be a directed set and $(f_i)_{i\in I}$ a monotone increasing net of upper functionals, pointwise bounded from above. Then $\bar{f}(x)=\sup_{i\in I}f_i(x)$ is a lower functional.
\item Let $f_1,\dots,f_r$ be lower functionals and $\theta\in\distributions([r])$. Then $\geometricmean(f_1,\dots,f_r)(x)=\prod_{i=1}^r f_i(x)^{\theta(i)}$ is a lower functional. Moreover, if $f_1,\dots,f_r$ are multiplicative then $\geometricmean(f_1,\dots,f_r)$ is also multiplicative.
\end{enumerate}
\end{proposition}
\begin{proof}
It is clear that all of these constructions preserve the property that $0$ and $1$ (in $S$) are mapped to $0$ and $1$ (in $\nonnegativereals$) as well as monotonicity. $\bar{g}(x)$ is subadditive since $\bar{g}(x+y)=\sup_{g\in G}g(x+y)\le\sup_{g\in G}(g(x)+g(y))\le\sup_{g\in G}g(x)+\sup_{g\in G}g(y)=\bar{g}(x)+\bar{g}(y)$. One shows in a similar way that $\bar{g}$ is submultiplicative, and $\ubar{f}$ is superadditive and supermultiplicative.

To see that $\ubar{g}$ is subadditive, for $x,y\in S$ and $\epsilon>0$ let $i,j\in I$ such that $g_i(x)\le\ubar{g}(x)+\epsilon$ and $g_j(y)\le\ubar{g}(y)+\epsilon$. Let $k$ be an upper bound of $i$ and $j$. Then
\begin{equation}
\ubar{g}(x+y)\le g_k(x+y)\le g_k(x)+g_k(y)\le\ubar{g}(x)+\ubar{g}(y)+2\epsilon.
\end{equation}
Since this is true for all $\epsilon>0$, we have $\ubar{g}(x+y)\le\ubar{g}(x)+\ubar{g}(y)$. The proofs of submultiplicativity of $\ubar{g}$ and superadditivity and supermultiplicativity of $\bar{f}$ are similar.

Since the geometric mean is multiplicative and monotone, $\geometricmean(f_1,\dots,f_r)$ inherits supermultiplicativity (multiplicativity). Superadditivity follows from the joint concavity and homogeneity of the geometric mean: for $x,y\in S$ we have
\begin{equation}
\begin{split}
\geometricmean(f_1(x+y),\dots,f_k(x+y))
 & \ge \geometricmean(f_1(x)+f_1(y),\dots,f_k(x)+f_k(y))  \\
 & = 2\geometricmean\left(\frac{f_1(x)+f_1(y)}{2},\dots,\frac{f_k(x)+f_k(y)}{2}\right)  \\
 & \ge 2\frac{\geometricmean(f_1(x),\dots,f_k(x))+\geometricmean(f_1(y),\dots,f_k(y))}{2}.
\end{split}
\end{equation}
\end{proof}

The functionals in \cref{ex:functionalexamples} are normalized in the sense that they send any natural number $n\in\naturals\subseteq S$ to itself. While this is not necessarily the case in general, it holds whenever a lower functional is less than an upper functional, a common situation and one which we are interested in:
\begin{proposition}\label{prop:uppernormalized}
Let $g$ be an upper functional. The following are equivalent:
\begin{enumerate}
\item\label{it:uppergreaterthansubrank} $g\ge\abstractsubrank$;
\item\label{it:uppergreaterthanlower} there exists a lower functional $f$ such that $g\ge f$;
\item\label{it:uppergreaterthann} $g(n)\ge n$ for all $n\in\naturals$;
\item\label{it:uppernormalized} $g(n)=n$ for all $n\in\naturals$.
\end{enumerate}
\end{proposition}
\begin{proof}
\ref{it:uppergreaterthansubrank}$\implies$\ref{it:uppergreaterthanlower} is clear since $\abstractsubrank$ is a lower functional.

\ref{it:uppergreaterthanlower}$\implies$\ref{it:uppergreaterthann} is true because $f(n)\ge n$ for any lower functional $f$ and $n\in\naturals$.

\ref{it:uppergreaterthann}$\implies$\ref{it:uppernormalized} is true because $g(n)\le n$ always holds for a upper functional $g$ and $n\in\naturals$.

\ref{it:uppernormalized}$\implies$\ref{it:uppergreaterthansubrank}: if $s\in S$, then $s\preorderge\abstractsubrank(s)$ by definition. It follows that $g(s)\ge g(\abstractsubrank(s))=\abstractsubrank(s)$ since $g$ is monotone and is the identity on natural numbers.
\end{proof}
The analogous statement for lower functionals is also true. We state it for completeness, but omit the proof which is very similar:
\begin{proposition}\label{prop:lowernormalized}
Let $f$ be a lower functional. The following are equivalent:
\begin{enumerate}
\item\label{it:lowerlessthanrank} $f\le\abstractrank$;
\item\label{it:lowerlessthanupper} there exists an upper functional $g$ such that $f\le g$;
\item\label{it:lowernlessthann} $f(n)\le n$ for all $n\in\naturals$;
\item\label{it:lowernormalized} $f(n)=n$ for all $n\in\naturals$.
\end{enumerate}
\end{proposition}
It is clear that the constructions in \cref{prop:functionalbasicconstructions} preserve the normalization.

Recall that the asymptotic rank and asymptotic subrank of tensor are upper and lower functionals that are multiplicative and additive when restricted to polynomials in a single tensor, with natural coefficients. In the following definition we generalize this property, allowing the natural numbers to be replaced with a larger subsemiring.
\begin{definition}
Let $S_0\subseteq S$ be a subsemiring. We call an upper (lower) functional \emph{$S_0$-spectral} if its restriction to $S_0$ is an element of $\Delta(S_0,\preorderle)$.

Let $g$ be an $S_0$-spectral upper functional. If $p=\sum_{i=0}^d a_iT^i\in S_0[T]$ is a polynomial (in the indeterminate $T$ and with coefficients from $S_0$), then we set $p_g=\sum_{i=0}^d g(a_i)T^i\in\nonnegativereals[T]$. We say that $g$ is \emph{$S_0$-regular}, if for all $x\in S$ and any polynomial $p\in S_0[T]$ we have $g(p(x))=p_g(g(x))$. In the special case when $S_0=\naturals$, we will also say \emph{regular} instead of $\naturals$-regular. Note that in this case $p=p_g$ are polynomials with natural coefficients.

Similarly, an $S_0$-spectral lower functional $f$ is \emph{$S_0$-regular} if $f(p(x))=p_f(f(x))$ for all $x\in S$ and $p\in S_0[T]$.
\end{definition}

With this terminology, the normalized functionals are precisely the $\naturals$-spectral ones. For example, a lower or upper functional is a spectral point iff it is $S$-spectral, which is in turn equivalent to $S$-regularity. In any semiring where $x\neq 0\implies x\preorderge 1$, the asymptotic rank $\abstractasymptoticrank(x)$ and the asymptotic subrank are $\abstractasymptoticsubrank(x)$ regular.

\begin{proposition}\label{prop:regularimpliesspectral}
If $g$ is $S_0$-regular, $s\in S$ and $\langle S_0,s\rangle$ denotes the subsemiring generated by $S_0$ and $s$, then $g$ is also $\langle S_0,s\rangle$-spectral.
\end{proposition}
\begin{proof}
The statement is a consequence of the fact that elements of $\langle S_0,s\rangle$ can be written as polynomials of the generator $s$ with coefficients in $S_0$, and that the map $p\mapsto p_g$ is a semiring homomorphism. In detail, let $x=p(s)$, $y=q(s)$, where $p$ and $q$ are polynomials with coefficients in $S_0$. For $r:=pq$ we have that $r_g=p_gq_g$, because $g$ is $S_0$-spectral, therefore
\begin{equation}
\begin{split}
g(xy)
 & = g(p(s)q(s))
 = g(r(s))
 = r_g(g(s))
 = p_g(g(s))1_g(g(s))  \\
 & = g(p(s))g(q(s))
 = g(x)g(y).
\end{split}
\end{equation}

Similarly, for the additivity we use that $r:=p+q$ implies $r_g=p_g+q_g$:
\begin{equation}
\begin{split}
g(x+y)
 & = g(p(s)+q(s))
 = g(r(s))
 = r_g(g(s))
 = p_g(g(s))+q_g(g(s))  \\
 & = g(p(s))+g(q(s))
 = g(x)+g(y).
\end{split}
\end{equation}
\end{proof}

\begin{remark}\label{rem:lowerequalsupper}
When $f$ is a lower functional and $g$ is an upper functional such that $f\le g$, then $S_0:=\setbuild{x\in S}{f(x)=g(x)}$ is a subsemiring. Indeed, these elements include the natural numbers by \cref{prop:uppernormalized,prop:lowernormalized} and if $f(x)=g(x)$ and $f(y)=g(y)$ then $g(x+y)\le g(x)+g(y)=f(x)+f(y)\le f(x+y)\le g(x+y)$ and $g(xy)\le g(x)g(y)=f(x)f(y)\le f(xy)\le g(xy)$. Both $f$ and $g$ are $S_0$-spectral, and their restrictions to $S_0$ are also the restrictions of at least one element of $\Delta(S,\preorderle)$ \cite[Corollary 2.18.]{zuiddam2018algebraic}.

For example, the support functionals with weight $\theta$ provide spectral points on the semiring of \emph{$\theta$-robust tensors} \cite[Section 4]{strassen1991degeneration}. Over the complex numbers, they arise as the restrictions of the quantum functionals \cite{christandl2021universal}.
\end{remark}

We state a converse in terms of the following generalization of the abstract rank and subrank. If $S_0\subseteq S$ is a subsemiring and $\varphi\in\Delta(S_0,\preorderle)$, then we set
\begin{align}
\abstractrank_\varphi(x) & = \inf\setbuild{\varphi(z)}{z\in S_0,z\preorderge x}  \\
\intertext{and}
\abstractsubrank_\varphi(x) & = \sup\setbuild{\varphi(z)}{z\in S_0,z\preorderle x}.
\end{align}
Note that $\abstractrank_\varphi$ is an upper functional and $\abstractsubrank_\varphi$ is a lower functional, and they both agree with $\varphi$ on $S_0$ (therefore both are $S_0$-spectral). Clearly any monotone extension of $\varphi$ is between $\abstractsubrank_\varphi$ and $\abstractrank_\varphi$.
\begin{proposition}\leavevmode
\begin{enumerate}
\item If $g$ is an $S_0$-spectral upper functional then $g\ge\abstractsubrank_{g|_{S_0}}$ and the two agree on $S_0$.
\item If $f$ is an $S_0$-spectral lower functional then $f\le\abstractrank_{f|_{S_0}}$ and the two agree on $S_0$.
\end{enumerate}
\end{proposition}
\begin{proof}
Since $g$ is monotone, $z\preorderge x$ implies $g(z)\le g(x)$, therefore $\abstractsubrank_{g|_{S_0}}(x)\le g(x)$. The second part is similar.
\end{proof}

\begin{proposition}
In the setting of \cref{prop:functionalbasicconstructions}, if $G$ (respectively $(g_i)_{i\in I}$, $F$, $(f_i)_{i\in I}$, ${f_1,\dots,f_r}$) consists of $S_0$-spectral upper (lower) functionals, and their restrictions to $S_0$ are all equal to $\varphi\in\Delta(S_0,\preorderle)$, then $\bar{g}$ (respectively $\ubar{g}$, $\ubar{f}$, $\bar{f}$, $\geometricmean(f_1,\dots,f_r)$) is also $S_0$-spectral and its restriction to $S_0$ is also $\varphi$.

If in addition every functional in $G$ (respectively $(g_i)_{i\in I}$, $F$, $(f_i)_{i\in I}$) is $S_0$-regular, then $\bar{g}$ (respectively $\ubar{g}$, $\ubar{f}$, $\bar{f}$) is also $S_0$-regular.
\end{proposition}
\begin{proof}
The first part is clear, since the infimum, supremum, and weighted geometric means of a constant family are all equal.

Suppose that $G$ is a set of $S_0$-regular upper functionals that restrict to $\varphi$ on $S_0$. Let $x\in S$. Then for all $g\in G$ and $p\in S_0[T]$ we have $g(p(x))=p_g(g(x))=p_{\bar{g}}(g(x))$. Since $p_{\bar{g}}$ is monotone increasing,
\begin{equation}
\bar{g}(p(x))
 = \sup_{g\in G}g(p(x))
 = \sup_{g\in G}p_{\bar{g}}(g(x))
 = p_{\bar{g}}\left(\sup_{g\in G}g(x)\right)
 = p_{\bar{g}}(\bar{g}(x)),
\end{equation}
therefore $\bar{g}$ is $S_0$-regular. The proofs for $\ubar{g}$, $\ubar{f}$, $\bar{f}$ are similar.
\end{proof}
\begin{remark}\label{rem:lowercounterexample}
The geometric mean of $S_0$-regular lower functionals, with equal restrictions to $S_0$, is in general not $S_0$-regular. For example, let $S=\positivereals^2\cup\{(0,0)\}$ with the componentwise operations and partial order. The functionals $(a,b)\mapsto a$ and $(a,b)\mapsto b$ are spectral points, therefore $\naturals$-regular lower functionals, but the map $(a,b)\mapsto \sqrt{ab}$ is not $\naturals$-regular (e.g. it sends $(1,1)$ to $1$, $(1,4)$ to $2$ and $(1,1)+(1,4)=(2,5)$ to $\sqrt{10}\neq 1+2$).
\end{remark}

\begin{example}
On the semiring of complex tensors with the restriction preorder, the quantum functionals $F^\theta$, parametrized with a distribution $\theta\in\distributions([k])$, form a closed set of spectral points. The asymptotic slice rank of a tensor $t$ is equal to $\min_{\theta\in\distributions([k])}F^\theta(t)$ \cite{christandl2021universal}, therefore the asymptotic slice rank is a regular lower functional.
\end{example}

Our next goal is to describe a construction that turns an $S_0$-spectral upper functional into an $S_0$-regular one.
\begin{definition}\label{def:abstractreg}
Let $g$ be an $S_0$-spectral upper functional. The $S_0$-regularization of $g$ is
\begin{equation}
\tilde{g}(x)=\inf_{\substack{p\in S_0[T]  \\  \deg p\ge 1}} p_g^{-1}(g(p(x))),
\end{equation}
where $x\in S$.
\end{definition}

\begin{remark}\label{rem:reggprop}\leavevmode
\begin{enumerate}
\item The identity is also a nonconstant polynomial, which implies $\tilde{g}\le g$.
\item If $g$ is $S_0$-regular then $p_g^{-1}(g(p(x))=p_g^{-1}(p_g(g(x)))=g(x)$ for all $x\in S$ and nonconstant $p\in S_0[T]$, therefore $\tilde{g}=g$.
\item If $g_1$ and $g_2$ are both $S_0$-spectral, $g_1|_{S_0}=g_2|_{S_0}$ and $g_1\le g_2$, then for all $x\in S$ and nonconstant $p\in S_0[T]$ we have $p_{g_1}^{-1}(g_1(p(x)))=p_{g_2}^{-1}(g_1(p(x)))\le p_{g_2}^{-1}(g_2(p(x)))$, therefore $\tilde{g}_1\le\tilde{g}_2$.
\item In particular, if $g$ is $\naturals$-spectral and $h$ is a spectral point, then $h\le g$ implies $h\le\tilde{g}$ (where $\tilde{g}$ is now the $\naturals$-regularization).
\end{enumerate}
\end{remark}

We will make use of an alternative expression for the $S_0$-regularization, which is similar to the definition of the functionals $\undertilde{\zeta}^\theta$ and $\undertilde{\zeta}_\theta$ in \cite[Section 7. (ii)]{strassen1991degeneration}.
\begin{proposition}\label{prop:twostepregularized}
The $S_0$-regularization $\tilde{g}$ of $g$ can also be expressed as
\begin{equation} \label{eq:regformulag}
\tilde{g}(x) =\lim_{n\to\infty}\sqrt[n]{\inf_{t\in S_0}\frac{g(tx^n)}{g(t)}}.
\end{equation}
\end{proposition}
\begin{proof}

The inequality
\begin{equation}\label{eq:regformulaineq1}
   \tilde{g}(x)\le \lim_{n\to\infty}\sqrt[n]{\inf_{t\in S_0}\frac{g(tx^n)}{g(t)}} 
\end{equation}
follows from the definition of $\tilde{g}$ by choosing the polynomials $tT^n$ (with $T$ the indeterminate), and replacing the infimum with a limit, which we can do by submultiplicativity.

For the converse inequality, let us denote the right hand side of \eqref{eq:regformulag} by $g'(x)$. The submultiplicativity of $g$ implies that that $g'\le g$. Since $g$ is $S_0$-spectral, it coincides with $g'$ on $S_0$, and therefore $g'$ is multiplicative on $S_0$ and $g'(n)=g(n)=n$ for all $n\in\naturals$.

Using the fact that $g$ is multiplicative on $S_0$, for $x\in S, u\in S_0$ we get
\begin{equation}\label{eq:subsemiringmultip}
\begin{split}
    g'(ux)&=\lim_{n\to\infty}\sqrt[n]{ \inf_{t\in S_0\backslash \{0\}} \frac{g(tu^nx^n)}{g(t)}}\\&= \lim_{n\to\infty}\sqrt[n]{g(u^n) \inf_{t\in S_0\backslash \{0\}} \frac{g(tu^nx^n)}{g(u^n)g(t)}}\\ &=
    g'(u)\lim_{n\to\infty}\sqrt[n]{ \inf_{t\in S_0\backslash \{0\}} \frac{g(tu^nx^n)}{g(u^nt)}} \\&\ge g'(u)g'(x).
\end{split} 
\end{equation}

Since $g$ is submultiplicative and $p_g^{-1}$ is monotone increasing for all $p\in S_0[T]$, the $S_0$-regularization may be written as
\begin{equation}
\begin{split}
\tilde{g}(x)
 & = \inf_p p_g^{-1}(g(p(x)))  \\
 & = \inf_p p_g^{-1}(\sqrt[n]{g(p(x)^n)}).
\end{split}
\end{equation}
We write the $n$-th power of a general polynomial $p\in S_0[T]$, evaluated at $x$, in the form 
\begin{equation}\label{eq:polynomialpower}
    p(x)^n=\sum_{Q\in\partitions[n]([0,d])}|T_Q^n|\prod_{i=0}^d(t_ix^i)^{nQ(i)},
\end{equation}
where $\distributions[n]([0,d])$ is the set of probability distributions on $[0,d]=\{0,\dots,d\}$ that assign multiples of $1/n$ to each point ($n$-types on $[0,d]$), and $T_Q^n$ is the set of strings in $[0,d]^n$ where each number $i$ occurs $nQ(i)$ times (type class).
\begin{equation}\label{eq:reggineq}
    \begin{split}
    \tilde{g}(x)&=
    \inf_p p_g^{-1}\bigg(\lim_{n\to\infty}\big(g(\sum_{Q\in\partitions[n]([0,d])}\lvert T_Q^n\rvert\prod_{i=0}^d(t_ix^i)^{nQ(i)})\big)^{\frac{1}{n}} \bigg)
    \\ 
    &\ge \inf_p p_g^{-1}\bigg(\lim_{n\to\infty}\big(\max_{Q\in\partitions[n]([0,d])} g(\lvert T_Q^n\rvert\prod_{i=0}^d(t_ix^i)^{nQ(i)})\big)^{\frac{1}{n}} \bigg)
    \\ &\ge 
    \inf_p p_g^{-1}\bigg(\lim_{n\to\infty}\big(\max_{Q\in\partitions[n]([0,d])} g'(\lvert T_Q^n\rvert\prod_{i=0}^d(t_ix^i)^{nQ(i)})\big)^{\frac{1}{n}}\bigg)\\
    &\ge\inf_p p_g^{-1}\bigg(\lim_{n\to\infty}\big(\max_{Q\in\partitions[n]([0,d])} g'(\lvert T_Q^n\rvert\prod_{i=0}^dt_i^{nQ(i)})
     g'(\prod_{i=0}^dx^{inQ(i)})  \big)^{\frac{1}{n}}\bigg)\\
    &=\inf_p p_g^{-1}\bigg(\lim_{n\to\infty}\bigg(\max_{Q\in\partitions[n]([0,d])} \lvert T_Q^n\rvert(\prod_{i=0}^d g'(t_i)^{nQ(i)})g'(x)^{\sum_{i=0}^dniQ(i)}\bigg)^{\frac{1}{n}}\bigg)\\ &\ge
    \inf_p p_g^{-1}\bigg(\lim_{n\to\infty}\bigg((n+1)^{-(d+1)}p_g(g'(x))^n\bigg)^{\frac{1}{n}}\bigg)\\[12pt]&=
    \inf_p p_g^{-1}(p_g(g'(x)))=g'(x),
\end{split}
\end{equation}
where the first inequality is due to the lower bound on the sum in \eqref{eq:polynomialpower} -- in the semiring preorder sense -- by only one of its elements\footnote{Note that an analogous upper bound -- i.e. taking the maximal element multiplied by the number of elements -- does not hold in general, which prevents us to perform a similar proof for lower functionals. In fact we know by \cref{rem:lowermaincounterexample} below that the analogous statement for lower functionals cannot hold.}. The third is by \eqref{eq:subsemiringmultip}, then the equality is due to the multiplicativity of $g'$ on $S_0$ and that $g'(x^i)=g'(x)^i$ for $i\in \naturals$ by the definition of $g'$. In the last inequality we bound the polynomial $p_g(g'(x))^n$, expanded in a similar way as in \eqref{eq:polynomialpower}, by the number of terms times the largest term, using that there are at most $(n+1)^{d+1}$ terms.
\end{proof}
If $S_0=\naturals$, then the infimum in \eqref{eq:regformulag} can be replaced with a limit $t\to\infty$ by subadditivity.

\begin{proposition}\label{prop:gregupper}
    $\tilde{g}$ is an upper functional.
    
\end{proposition}
\begin{proof}
Since $\tilde{g}$ agrees with $g$ on $S_0$, we have $\tilde{g}(0)=0$ and $\tilde{g}(1)=1$. It follows directly from the definition that $\tilde{g}$ is monotone, i.e. $\tilde{g}(x)\le \tilde{g}(y)$ when $x\preorderle y$. 
From \cref{prop:twostepregularized} the submultiplicativity of $\tilde{g}$ follows as
\begin{equation}\label{eq:gregmultip}
\begin{split}
    \tilde{g}(xy)&=\lim_{n\to\infty}\sqrt[n]{\inf_{t\in S_0\backslash \{0\}} \frac{g(tx^ny^n)}{g(t)}}\\
    &=
    \lim_{n\to\infty}\sqrt[n]{\inf_{t\in S_0\backslash \{0\}} \inf_{u\in S_0\backslash \{0\}} \frac{g(tux^ny^n)}{g(tu)}}\\[3pt]
    &\leq \lim_{n\to\infty}\sqrt[n]{\inf_{t\in S_0\backslash \{0\}} \frac{g(tx^n)}{g(t)} \inf_{u\in S_0\backslash \{0\}} \frac{g(uy^n)}{g(u)}}\\[7pt]
    &= \tilde{g}(x)\tilde{g}(y).
\end{split}
\end{equation}

The subadditivity follows as:
\begin{equation}\label{eq:gregsubadd}
\begin{split}
    \tilde{g}(x+y)&= \lim_{n\to\infty}\sqrt[n]{\inf_{t\in S_0\backslash \{0\}} \frac{g(t(x+y)^n)}{g(t)}} \\ 
    &= \lim_{n\to\infty}\sqrt[n]{\inf_{t\in S_0\backslash \{0\}}\frac{g\big(t\sum_{i=0}^n {\binom{n}{i}}x^iy^{n-i}\big)}{g(t)}} \\ 
    &\leq \lim_{n\to\infty}\sqrt[n]{\inf_{t_1\in S_0\backslash \{0\}}\dots
    \inf_{t_n\in S_0\backslash \{0\}}
    \sum_{i=0}^n {\binom{n}{i}} \frac{g\big(t_1\dots t_nx^iy^{n-i}\big)}{g(t_1\dots t_n)}} \\ 
    &\leq \lim_{n\to\infty}\sqrt[n]{\inf_{t_1\in S_0\backslash \{0\}}\dots
    \inf_{t_n\in S_0\backslash \{0\}}
    \sum_{i=0}^n {\binom{n}{i}} \frac{g\big(t_ix^iy^{n-i}\big)}{g(t_i)}} \\ 
    &=  \lim_{n\to\infty}\sqrt[n]{
    \sum_{i=0}^n {\binom{n}{i}} \inf_{t_i\in S_0\backslash \{0\}} \frac{g\big(t_ix^iy^{n-i}\big)}{g(t_i)}} \\ 
    &\le  \lim_{n\to\infty}\sqrt[n]{
    \sum_{i=0}^n {\binom{n}{i}} \left(\inf_{t\in S_0\backslash \{0\}} \frac{g\big(tx^i\big)}{g(t)}\right)
    \left(\inf_{u\in S_0\backslash \{0\}}\frac{g\big(uy^{n-i}\big)}{g(u)}\right)}, 
\end{split}
\end{equation}
where in the last inequality we do the same conversions as in \eqref{eq:gregmultip}.
Now we use that $\lim_{i\to\infty} (\inf_{t\in S_0\backslash \{0\}} \frac{g(tx^i)}{g(t)})^{\frac{1}{i}}=\tilde{g}(x)$
\begin{equation}\label{eq:ggregepsilonclose}
\begin{split}
\implies\quad  \forall\epsilon>0,  \exists c>0, \forall i\in\naturals: \quad &\inf_{t\in S_0\backslash \{0\}} \frac{g\big(tx^i\big)}{g(t)}\leq c (\tilde{g}(x)+\epsilon)^i,\\
&   \inf_{u\in S_0\backslash \{0\}} \frac{g\big(uy^{n-i}\big)}{g(u)}\leq c (\tilde{g}(y)+\epsilon)^{n-i}
\end{split}
\end{equation}
so 
\begin{equation}
\begin{split}
    \sum_{i=0}^n {\binom{n}{i}} \Big(\inf_{t\in S_0\backslash \{0\}} \frac{g\big(tx^i\big)}{g(t)}\Big)
    \Big(\inf_{u\in S_0\backslash \{0\}}\frac{g\big(uy^{n-i}\big)}{g(u)}\Big)
     & \le c^2\sum_{i=0}^n {\binom{n}{i}}\bigg(\tilde{g}(x)+\epsilon\bigg)^i\bigg(\tilde{g}(y)+\epsilon\bigg)^{n-i}  \\
     & =  c^2(\tilde{g}(x)+\tilde{g}(y)+2\epsilon)^n.
\end{split}
\end{equation}
The statement follows by substituting this to \eqref{eq:gregsubadd}.
\end{proof}

\begin{proposition}
     The $S_0$-regularization $\tilde{g}$ of $g$ is $S_0$-regular.
\end{proposition}

\begin{proof}
By \cref{prop:gregupper} we have to show only $\tilde{g}(q(x))\ge q_g(\tilde{g}(x))$ for all $q\in S_0[T]$. Let $q\in S_0[T]$.
\begin{equation}\label{eq:gregsuper}
\begin{split}
\tilde{g}(q(s))
 & = \inf_{p\in S_0[T]}  (q_g\circ q_g^{-1})p_g^{-1}\bigg(g\big((p\circ q)(s)\big)\bigg)  \\
 & = \inf_{p\in S_0[T]}  q_g \bigg((p_g\circ q_g)^{-1}g\big((p\circ q)(s)\big)\bigg)  \\
 & = q_g\bigg(\inf_{\substack{r=p\circ q\\ p\in S_0[T]}} r_g^{-1}(g(r(s))) \bigg)  \\
 & \ge q_g\bigg(\inf_{r\in S_0[T]} r_g^{-1}(g(r(s))) \bigg)=q_g(\tilde{g}(s)).
\end{split}
\end{equation}
From the submultiplicativity and subadditivity the converse inequality also follows, so in fact in \eqref{eq:gregsuper} we have an equality.

\end{proof}

\begin{remark}\label{rem:successivereg}\leavevmode
\begin{enumerate}
\item By successive application of regularization on a $\naturals$-spectral upper functional one can construct an $S_0$-regular upper functional for any finitely generated subsemiring $S_0$.

\item An interesting consequence is that any $\naturals$-spectral and multiplicative upper functional is also additive, i.e. a spectral point. To see this, one observes that \eqref{eq:regformulag} implies that $\tilde{g}=g$, therefore a multiplicative upper functional is $S_0$-spectral iff it is $S_0$-regular. Using \cref{prop:regularimpliesspectral} we see that if $g$ is multiplicative and $\naturals$-spectral, then it is also $\langle x\rangle$-spectral, therefore also $\langle x\rangle$-regular for all $x\in S$, which implies additivity.

In contrast, a normalized multiplicative \emph{lower} functional need not be a spectral point (see \cref{rem:lowercounterexample}).
\end{enumerate}
\end{remark}

If a lower functional $f$ is a lower bound on every element of the subset $H\subseteq\Delta(S,\preorderle)$, then the upper functional $g(x)=\sup_{h\in H}h(x)$ also satisfies $f\le g$. However, the inequality $f\le g$ does not even guarantee the existence of a single spectral point between $f$ and $g$. In the next definition we introduce a stronger property which, as we will show below, implies that the regularization of $g$ is the maximum of the spectral points between $f$ and $g$.
\begin{definition}
Let $f$ be a lower functional and $g$ an upper functional. We say that $g$ is \emph{$f$-supermultiplicative} if for all $x,y\in S$ the inequality $f(x)g(y)\le g(xy)$ holds.
\end{definition}
Note that if $g$ is $f$-supermultiplicative then $f\le g$.

\begin{example}\label{ex:fsupermultiplcative}\leavevmode
\begin{enumerate}
\item On the semiring of tensor classes over an arbitrary field, the upper support functional $\zeta^\theta$ is $\zeta_\theta$-supermultiplicative \cite[Corollary 4.2.]{strassen1991degeneration}.
\item On the semiring of graphs (with disjoint union, strong product, and the cohomomorphism preorder) \cite{zuiddam2019asymptotic}, the fractional clique cover number $\complement{\chi}_f$ is a spectral point and the clique cover number $\complement{\chi}$ is an upper functional, and $\complement{\chi}$ is $\complement{\chi}_f$-supermultiplicative \cite[Proposition 3.4.4]{scheinerman2011fractional}.
\item If $f$ is a lower functional and $h\in\Delta(S,\preorderle)$ such that $f\le h$, then $f(x)h(y)\le h(x)h(y)=h(xy)$, therefore $h$ is $f$-supermultiplicative.
\item If $g$ is $S_0$-regular then it is also $\abstractsubrank_{g|_{S_0}}$-supermultiplicative. To see this, let $x,y\in S$ and $z\in S_0$ and suppose that $z\preorderle x$. Then $xy\preorderge zy$, therefore $g(xy)\ge g(zy)=g(z)g(y)$. Taking the supremum over all $z\preorderle x$ gives $g(xy)\ge \abstractsubrank_{g|_{S_0}}(x)g(y)$.
\item In particular, if $g$ is regular then it is also $Q$-supermultiplicative. In fact, $g(nx)=nx$ for all $n\in\naturals$ and $x\in S$ is sufficient for this (which in turn is implied by additivity).
\end{enumerate}
\end{example}

$f$-supermultiplicativity is preserved by the constructions discussed in \cref{prop:functionalbasicconstructions,def:abstractreg}:
\begin{proposition}\label{prop:fsupermultiplicativeconstructions}
Let $f$ be a lower functional.
\begin{enumerate}
\item\label{it:fsupermultiplicativesupremum} The supremum of a nonempty set of $f$-supermultiplicative upper functionals is also $f$-supermultiplicative.
\item\label{it:fsupermultiplicativeinfimum} The infimum of a monotone decreasing net of $f$-supermultiplicative upper functionals is also $f$-supermultiplicative.
\item\label{it:fsupermultiplicativeregularization} The $S_0$-regularization of an $S_0$-spectral $f$-supermultiplicative upper functional is also $f$-supermultiplicative.
\end{enumerate}
\end{proposition}
\begin{proof}
\ref{it:fsupermultiplicativesupremum}:
Let $G$ be a nonempty set of $f$-supermultiplicative upper functionals. $\bar{g}(x)=\sup_{g\in G}g(x)$ is an upper functional by \cref{prop:functionalbasicconstructions}. Let $x,y\in S$ and $\epsilon>0$, and choose $g\in G$ such that $g(y)\ge\bar{g}(y)-\epsilon$. Using $f$-supermultiplicativity of $g$, we have
\begin{equation}
f(x)\bar{g}(y)
 \le f(x)(g(y)+\epsilon)
 \le g(xy)+\epsilon f(x)
 \le \bar{g}(xy)+\epsilon f(x).
\end{equation}
Since this is true for every $\epsilon>0$, the inequality $f(x)\bar{g}(y)\le\bar{g}(xy)$ also holds.

\ref{it:fsupermultiplicativeinfimum}:
Let $(I,\le)$ and $(g_i)_{i\in I}$ be as in \cref{prop:functionalbasicconstructions}, so that $\ubar{g}(x)=\inf_{i\in I}g_i(x)$ is an upper functional. Suppose that $g_i$ is $f$-supermultiplicative for all $i\in I$. Let $x,y\in S$ and $\epsilon>0$, and choose $i\in I$ such that $g_i(xy)\le\ubar{g}(xy)+\epsilon$. Then
\begin{equation}
f(x)\ubar{g}(y)
 \le f(x)g_i(y)
 \le g_i(xy)
 \le \ubar{g}(xy)+\epsilon.
\end{equation}
Since this is true for every $\epsilon>0$, the inequality $f(x)\ubar{g}(y)\le\ubar{g}(xy)$ also holds.

\ref{it:fsupermultiplicativeregularization}:
Let $g$ be $S_0$-spectral and $f$-supermultiplicative, and let $\tilde{g}$ be its $S_0$-regularization. Using the characterization in \cref{prop:twostepregularized} we obtain for $x,y\in S$ the inequality
\begin{equation}
\tilde{g}(xy)
 = \lim_{n\to\infty}\sqrt[n]{\inf_{t\in S_0}\frac{g(tx^ny^n)}{g(t)}}
 \ge \lim_{n\to\infty}\sqrt[n]{\inf_{t\in S_0}\frac{g(tx^n)f(y^n)}{g(t)}}
 \ge \tilde{g}(x)f(y),
\end{equation}
therefore $\tilde{g}$ is $f$-supermultiplicative.
\end{proof}

\begin{corollary}
If $f$ is a lower functional, $H\subseteq\Delta(S,\preorderle)$ is closed and $f\le h$ for all $h\in H$, then $\bar{h}(x)=\max_{h\in H}h(x)$ is an $f$-supermultiplicative regular upper functional.
\end{corollary}
The main result of this section is the following converse.
\begin{theorem}\label{thm:reguppermaximum}
Let $(S,\preorderle)$ be a semiring with a Strassen preorder, $F$ a set of lower functionals, and $g$ an $\naturals$-spectral upper functional that is $f$-supermultiplicative for all $f\in F$. Let $\tilde{g}$ be the regularization of $g$ and $H=\setbuild{h\in\Delta(S,\preorderle)}{h\le g,\forall f\in F:f\le h}$. Then $H$ is nonempty and for all $x\in S$ the equality
\begin{equation}
\tilde{g}(x)=\max_{h\in H}h(x)
\end{equation}
holds.
\end{theorem}
\begin{proof}
We may assume that $F$ is nonempty, otherwise we can take $F=\{Q\}$ and replace $g$ with $\tilde{g}$ if necessary.

For $z\in S$ let $\mathcal{U}_z$ denote the set of upper functionals that are $f$-supermultiplicative for all $f\in F$, pointwise less than $g$, and agree with $\tilde{g}$ on $\langle z\rangle$. Note in particular that $\mathcal{U}_z\cap\Delta(S,\preorderle)\subseteq H$. Equip $\mathcal{U}_z$ with the pointwise partial order.

$\tilde{g}\in\mathcal{U}_z$ by \cref{rem:reggprop,prop:fsupermultiplicativeconstructions}, therefore $\mathcal{U}_z$ is nonempty. If $C$ is a chain in $\mathcal{U}_z$, then its pointwise infimum is also in $\mathcal{U}_z$ by \cref{prop:fsupermultiplicativeconstructions}, and is a lower bound on $C$. By Zorn's lemma, $\mathcal{U}_z$ contains at least one minimal element. Let one of these be $h$.

We claim that $h$ is an element of $H$. By the definition of $\mathcal{U}_z$, $h$ is a $\langle z\rangle$-spectral upper functional between $f$ and $g$ for all $f\in F$, therefore we need to show that $h$ is additive and multiplicative. The $\langle z\rangle$-regularization of $h$ is in $\mathcal{U}_z$, therefore equals $h$ by minimality. Let $x,y\in S$. $h$ is $\langle z,x\rangle$-spectral by \cref{prop:regularimpliesspectral}, therefore we can consider its $\langle z,x\rangle$-regularization, which is also in $\mathcal{U}_z$ by \cref{prop:fsupermultiplicativeconstructions}, therefore agrees with $h$ by minimality. Thus $h$ is $\langle z,x\rangle$-regular, which implies that it is also $\langle z,x,y\rangle$-spectral. It follows that $h(x+y)=h(x)+h(y)$ and $h(xy)=h(x)h(y)$. As $x$ and $y$ were arbitrary, this means that $h\in\Delta(S,\preorderle)$, therefore also $h\in H$. This proves that $H$ is nonempty and $\tilde{g}(z)\le\max_{h\in H}h(z)$.

On the other hand, $h\le g$ implies $h\le\tilde{g}$ by \cref{rem:reggprop}, therefore $\tilde{g}(z)\ge\max_{h\in H}h(z)$.
\end{proof}
We now give a different proof of the statement by a reduction to Strassen's theorem, in particular to the assertion that the asymptotic spectrum of a semiring with a Strassen preorder is nonempty, and a corollary stating that a spectral point of a subsemiring can be extended to an element of $\Delta(S,\preorderle)$ \cite[Corollary 2.18.]{zuiddam2018algebraic}.
\begin{proof}[Second proof of \cref{thm:reguppermaximum}]
The idea is to equip $S$ with a modified preorder that is also a Strassen preorder, and the asymptotic spectrum with respect to the new preorder is a subset of $H$, which will show that $H$ is nonempty. To this end, consider the relation
\begin{equation}
R = \setbuild{(x,n)\in S^2}{n\in\naturals,g(x)\le n} \cup \setbuild{(n,x)\in S^2}{n\in\naturals,n\le\smash{\sup_{f\in F}f(x)}},
\end{equation}
and the preorder $\preorderle_R$ defined as $a\preorderle_R b$ iff there is an $n\in\naturals$ and $s_1,\dots,s_n,x_1,\dots,x_n,\\y_1,\dots,y_n\in S$ such that $(x_i,y_i)\in R$ for all $i=1,\dots,n$ and the inequality
\begin{equation}
a+s_1y_1+\dots+s_ny_n\preorderle b+s_1x_1+\dots+s_nx_n
\end{equation}
holds \cite[Definition 4]{vrana2021generalization}. Then $(S,\preorderle_R)$ is a preordered semiring, and $a\preorderle b$ implies $a\preorderle_R b$, and also $(a,b)\in R$ implies $a\preorderle_R b$ \cite[Lemma 7]{vrana2021generalization}. In particular, $\Delta(S,\preorderle_R)\subseteq\Delta(S,\preorderle)$.

We show that $\Delta(S,\preorderle_R)\subseteq H$. Let $h\in\Delta(S,\preorderle_R)$ and $f\in F$. For all $x\in S$ and $r\in\naturals$ we have $\lfloor rf(x)\rfloor\le\lfloor f(rx)\rfloor\le f(rx)$, therefore $(\lfloor rf(x)\rfloor,rx)\in R$, which implies $\lfloor rf(x)\rfloor\preorderle_R rx$. Applying $h$ to both sides yields $\lfloor rf(x)\rfloor\le h(rx)=rh(x)$, which implies
\begin{equation}
h(x)\ge\lim_{r\to\infty}\frac{\lfloor rf(x)\rfloor}{r}=f(x).
\end{equation}
Analogously, the inequality $g(rx)\le\lceil g(rx)\rceil\le\lceil rg(x)\rceil$ shows that $(rx,\lceil rg(x)\rceil)\in R$, which implies $rx\preorderle_R\lceil rg(x)\rceil$, therefore
\begin{equation}
h(x)\le\lim_{r\to\infty}\frac{\lceil rg(x)\rceil}{r}=g(x).
\end{equation}

Next we show that $\preorderle_R$ is a Strassen preorder. The only property that needs to be verified is that the homomorphism $\naturals\to S$ is order reflecting with respect to $\preorderle_R$. Let $m_1,m_2\in\naturals$ such that $m_1\preorderle_R m_2$. By the definition of $m_1\preccurlyeq_R m_2$ there exist $s_i, x_i\in S$ and $n_i\in\naturals$ ($i\in [r_1+r_2]$) and $f_i\in F$ ($i\in [r_1+1,r_2]$), such that $n_i\ge g(x_i)$ for $i\in [r_1]$, $f_i(x_i)\ge n_i$ for $i\in [r_1+1,r_2]$ and 
\begin{equation}\label{eq:newrelation}
    m_1 +\sum_{i=1}^{r_1}s_in_i + \sum_{i=r_1+1}^{r_1+r_2}s_ix_i \preorderle m_2 +\sum_{i=1}^{r_1}s_ix_i + \sum_{i=r_1+1}^{r_1+r_2}s_in_i.
\end{equation}
By \cref{rem:successivereg,prop:fsupermultiplicativeconstructions} there exists an upper functional $g'$ that is multiplicative and additive on $\langle x_1,\dots,x_{(r_1+r_2)} \rangle$, and satisfies $g'\le g$ as well as $f$-supermultiplicativity for all $f\in F$. Using its monotonicity we get
\begin{equation}\label{eq:actwithgreg}
\begin{split}
   m_1 +\sum_{i=1}^{r_1}g'(s_i)n_i + \sum_{i=r_1+1}^{r_1+r_2}g'(s_i)g'(x_i) &\le m_2 +\sum_{i=1}^{r_1}g'(s_i)g'(x_i) + \sum_{i=r_1+1}^{r_1+r_2}g'(s_i)n_i\\
      \sum_{i=1}^{r_1}g'(s_i)(n_i-g'(x_i)) + \sum_{i=r_1+1}^{r_1+r_2}g'(s_i)(g'(x_i)-n_i) &\le m_2 - m_1
\end{split}
\end{equation}
The first term on the left hand side is positive since $n_i\ge g(x_i)\ge g'(x_i)$, the second term is positive because $g'(x_i)\ge f(x_i)\ge n_i$ by $f_i$-multiplicativity of $g'$, so $0\le m_2-m_1$, i.e. $m_1\le m_2$. This shows that the preorder is indeed Strassen, therefore $\Delta(S,\preorderle_R)\neq\emptyset$ \cite[Theorem 2.4]{strassen1988asymptotic}.

As in the first proof, the inequality $\tilde{g}(z)\le\max_{h\in H}h(z)$ is clear. For the reverse inequality first we show that $\tilde{g}$ is $\preorderle_R$-monotone. Let $m_1,m_2\in S$ and $m_1\preorderle_R m_2$. Then by definition \eqref{eq:newrelation} can be written again. Now we use that by \eqref{rem:successivereg} there exists an upper functional $g'\le \tilde{g}$ which is $\langle x_1,\dots,x_{(r_1+r_2)} \rangle$-spectral and $g'(m_2)=\tilde{g}(m_2)$ (We start the iterative regularizations on the semiring $\langle m_2\rangle$). Applying $\tilde{g}$ to \eqref{eq:newrelation} we get a similar inequality as in \eqref{eq:actwithgreg} but with $\tilde{g}(m_2)-\tilde{g}(m_1)$ on the right hand side, which lower bounds $g'(m_2)-g'(m_1)$.

$\tilde{g}$ is regular and $\preorderle_R$-monotone, so it restricts to an element of the spectrum defined by the relation $\preorderle_R$ on any subsemiring of the form $\langle z\rangle$, where $z\in S$. By \cite[Corollary 2.18.]{zuiddam2018algebraic} there exists an extension of this in $\Delta(S,\preorderle_R)$, i.e. for $z\in S$ there exists $h_z\in\Delta(S,\preorderle_R)$ such that $h_z(z)=\tilde{g}(z)$. This implies $\tilde{g}(z)\ge\max_{h\in H}h(z)$.\end{proof}

\begin{corollary}\label{cor:S0reguppermaximum}
If $S_0$ is a subsemiring of $S$ and in addition to the conditions of \cref{thm:reguppermaximum}, $g$ is $S_0$-spectral and $\tilde{g}$ denotes its $S_0$-regularization, then
\begin{equation}
H:=\setbuild{h\in\Delta(S,\preorderle)}{\left.h\right\rvert_{S_0}=\left.g\right\rvert_{S_0},\forall f\in F:f\le h\le g}
\end{equation}
is nonempty and for all $x\in S$ the equality
\begin{equation}
\tilde{g}(x)=\max_{h\in H}h(x)
\end{equation}
holds.
\end{corollary}
\begin{proof}
\Cref{thm:reguppermaximum}, applied to $\tilde{g}$ and $F'=F\cup\{\abstractsubrank_{\left.\tilde{g}\right\vert_{S_0}}\}$ (possible by \cref{ex:fsupermultiplcative}), states that $H'=\setbuild{h\in\Delta(S,\preorderle)}{h\le \tilde{g},\forall f\in F':f\le h}$ is nonempty and $\tilde{g}(x)=\max_{h\in H'}h(x)$. On $S_0$ the lower bound $\abstractsubrank_{\left.\tilde{g}\right\vert_{S_0}}\le h$ matches the upper bound $h\le\tilde{g}$, therefore $H'=H$.
\end{proof}

\Cref{thm:reguppermaximum} can be used to give a new proof of the extension property of spectral points from subsemirings (\cite[Corollary 2.18.]{zuiddam2018algebraic}; see also \cite[7.17. Theorem]{fritz2021abstract} and \cite[Proposition 3]{vrana2021generalization} for semirings with more general preorders). Note that the second proof uses this property, but the first proof does not, thus the reasoning is not circular.
\begin{corollary}
Let $S_0\subseteq S$ be a subsemiring and $\varphi\in\Delta(S_0,\preorderle)$. Then there exists an element $h\in\Delta(S,\preorderle)$ such that $\varphi=h|_{S_0}$.
\end{corollary}
\begin{proof}
Choose $g=\abstractrank_{\left.\tilde{g}\right\vert_{S_0}}$ in \cref{cor:S0reguppermaximum}.
\end{proof}

A result similar to part of \cref{thm:reguppermaximum} was announced already in \cite[Section 7, (ii)]{strassen1991degeneration} in the special case when $g$ is the upper support functional, but we have not been able to locate a proof in the literature. The precise statement is that there is a unique minimal compact subset $\mathsf{Z}^\theta$ of the asymptotic spectrum of tensors such that $\tilde{\zeta}^\theta(x)=\max_{h\in \mathsf{Z}^\theta}h(x)$. Since $\zeta^\theta$ is $\zeta_\theta$-supermultiplicative, \cref{thm:reguppermaximum} implies a variant of this, where $\mathsf{Z}^\theta$ is replaced with the set of spectral points between the upper and lower support functionals, in particular showing that there exist spectral points between $\zeta_\theta$ and $\zeta^\theta$ over any field. This was previously not known except over the field of complex numbers, in which case the quantum functional $F^\theta$ is an explicit spectral point between the two \cite{christandl2021universal}.

\begin{remark}\label{rem:lowermaincounterexample}
The roles of upper and lower functionals in the results presented in this section are not symmetric, and it would be interesting to see how a counterpart of \cref{thm:reguppermaximum} for lower functionals could look like. There seems to be a fundamental asymmetry between upper and lower functionals: if $g$ is an $\naturals$-spectral upper functional, then there exist spectral points $h$ satisfying $h\le g$; however, there are $\naturals$-spectral lower functionals such that no spectral point $h$ satisfies $h\ge f$.

To see this, let $S=\positivereals^2\cup\{(0,0)\}$ with the pointwise operations and preorder $\preorderle$ (as in \cref{rem:lowercounterexample}). $\Delta(S,\preorderle)$ consists of exactly two elements, given by $h_1(a,b)=a$ and $h_2(a,b)=b$. The map $f(a,b)=\sqrt{ab}=\sqrt{h_1(a,b)h_2(a,b)}$ is a lower functional (see \cref{prop:functionalbasicconstructions}), but $f\not\le h_1$ and $f\not\le h_2$.

A key difference is that if we replace upper functionals with lower ones and the infimum with a supremum (in \cref{def:abstractreg} and thereafter) then \cref{prop:twostepregularized} becomes false. Continuing the previous example, since $f$ is multiplicative, we have
\begin{equation}
\lim_{n\to\infty}\sqrt[n]{\sup_{t\in\naturals}\frac{f(tx^n)}{f(t)}}=f(x),
\end{equation}
but
\begin{equation}
\sup_{\substack{p\in\naturals[T]  \\  \deg p\ge 1}}p_f^{-1}(f(p(a,b))=\max\{a,b\},
\end{equation}
which can be seen from the upper bound $f(p(a,b))=f(p(a),p(b))\le\max\{p(a),p(b)\}$ and the lower bound given by the polynomials $p=M+T^n$ with $M\to\infty$ and then $n\to\infty$.
\end{remark}

\section{Upper functionals from families of observables}\label{sec:upper}

From now on we will work with the preordered semiring of tensors and LOCC transformations. It is known that any spectral point $h$ has an exponent $\alpha\in[0,1]$ such that $h(\sqrt{p}\psi)=p^\alpha h(\psi)$ for all $p\ge 0$ and $\psi\in\mathcal{H}$ \cite[Theorem 3]{jensen2019asymptotic}. This is in general not true for upper or lower functionals: for instance, $g(\psi):=\max\{\tensorrank(\psi),\norm{\psi}^2\}$ is an upper functional by \cref{prop:functionalbasicconstructions}, but $p\mapsto g(\sqrt{p}\psi)$ is constant for small $p$ and proportional to $p$ for large arguments. We will focus on special functionals that do have a well-defined scaling behavior.
\begin{definition}
Let $k\in\naturals$, $k\ge 2$ and $\alpha\in(0,1]$. We will say that a functional $g$ on the semiring of order-$k$ tensors and LOCC transformations is an \emph{upper functional of order $\alpha$} if $g(1)=1$ and it satisfies the following properties:
\begin{enumerate}
\item\label[property]{it:upperscaling} $g(\sqrt{p}\ket{\psi})=p^\alpha g(\ket{\psi})$.
\item\label[property]{it:uppersubmultiplicative} $g(\ket{\psi}\otimes\ket{\varphi})\le g(\ket{\psi})g(\ket{\varphi})$
\item\label[property]{it:uppersubadditive} $g(\ket{\psi}\oplus\ket{\varphi})\le g(\ket{\psi})+g(\ket{\varphi})$
\item\label[property]{it:upperlocalprojection} 
$g(\ket{\psi})^{1/\alpha}\ge g(\Pi\ket{\psi})^{1/\alpha}+g((I-\Pi)\ket{\psi})^{1/\alpha}$.
\end{enumerate}
for all $p>0$, $\psi\in\mathcal{H}=\mathcal{H}_1\otimes\cdots\otimes\mathcal{H}_k$, $\varphi\in\mathcal{K}=\mathcal{K}_1\otimes\cdots\otimes\mathcal{K}_k$ and projection $\Pi\in\boundeds(\mathcal{H}_j)\subseteq\boundeds(\mathcal{H})$, for all $j\in[k]$.

Likewise, a functional $f$ is a \emph{lower functional of order $\alpha$} if it satisfies similar properties with the inequalities in \cref{it:uppersubmultiplicative,it:uppersubadditive} reversed (but not in \cref{it:upperlocalprojection}).
\end{definition}
Upper (lower) functionals of order $\alpha$ are upper (lower) functionals in the sense of \cref{def:abstractlowerupper}: the inequalities connecting the functionals and the operations are identical in the two definitions, while \cref{it:upperscaling,it:upperlocalprojection} imply that the functionals are monotone \cite[Theorem 3]{jensen2019asymptotic} (the statement there is concerned with spectral points, but the proof of this direction does not use multiplicativity or additivity).

In this section we describe a construction of upper functionals, in which operators on symmetric powers play a central role. Before we proceed, recall that there exist (up to phase) unique equivariant isometries
\begin{align}
    W_{\mathcal{H},m,n} & : \symmetricpower[m+n](\mathcal{H})\to\symmetricpower[m](\mathcal{H})\otimes\symmetricpower[n](\mathcal{K})  \\
    W_{\mathcal{H},\mathcal{K},n} & : \symmetricpower[n](\mathcal{H})\otimes\symmetricpower[n](\mathcal{K})\to\symmetricpower[n](\mathcal{H}\otimes\mathcal{K})  \\
    W_{\mathcal{H},\mathcal{K},m,n} & : \symmetricpower[m](\mathcal{H})\otimes\symmetricpower[n](\mathcal{K})\to\symmetricpower[m+n](\mathcal{H}\oplus\mathcal{K})
\end{align}
with respect to the groups $U(\mathcal{H})$ and $U(\mathcal{H})\times U(\mathcal{K})$, respectively. For brevity, when $X\in\boundeds(\symmetricpower[m](\mathcal{H})\otimes\symmetricpower[n](\mathcal{K}))$, we will use the notation
\begin{equation}
    \left.X\right|_{\symmetricpower[m+n](\mathcal{H})}=W_{\mathcal{H},m,n}^*XW_{\mathcal{H},m,n},
\end{equation}
and similarly in the other cases. Note in particular that (choosing the phase suitably)
\begin{equation}
W_{\mathcal{H},\mathcal{K},m,n}\left(\psi^{\otimes m}\otimes\varphi^{\otimes n}\right)=\frac{1}{\sqrt{m!n!(m+n)!}}\sum_{\sigma\in S_{m+n}}\sigma\cdot(\psi^{\otimes m}\otimes\varphi^{\otimes n})
\end{equation}
and therefore
\begin{equation}\label{eq:directsumpowerrestriction}
W_{\mathcal{H},\mathcal{K},m,n}^*\left((\psi\oplus\varphi)^{\otimes(m+n)}\right)=\sqrt{\binom{m}{m+n}}\psi^{\otimes m}\otimes\varphi^{\otimes n}.
\end{equation}

Our construction of upper functionals will be based on families of observables $A_{\mathcal{H},n}\in\boundeds(\symmetricpower[n](\mathcal{H}))$ indexed by $k$-partite Hilbert spaces $\mathcal{H}$ and natural numbers $n\in\naturals$, subject to the following axioms with some $\alpha\in(0,1)$:
\begin{enumerate}[({O}1)]
\item\label[property]{it:observableisometry} if $U_i:\mathcal{H}_i\to\mathcal{K}_i$ are isometries then $(U_1^*\otimes\dots\otimes U_k^*)^{\otimes n}A_{\mathcal{K},n}(U_1\otimes\dots\otimes U_k)^{\otimes n}=A_{\mathcal{H},n}$
\item\label[property]{it:observablebound} there exist constants $c_\mathcal{H}$ depending only on $\mathcal{H}$ such that $I\le A_{\mathcal{H},n}\le c_\mathcal{H}^n I$
\item\label[property]{it:observablesupermultiplicative} $\left.\left(A^{\frac{1-\alpha}{\alpha}}_{\mathcal{H},m}\otimes A^{\frac{1-\alpha}{\alpha}}_{\mathcal{H},n}\right)\right|_{\symmetricpower[n+m](\mathcal{H})}\le A^{\frac{1-\alpha}{\alpha}}_{\mathcal{H},m+n}$
\item\label[property]{it:observabletensorproduct} $\left.A^{\frac{1-\alpha}{\alpha}}_{\mathcal{H}\otimes\mathcal{K},n}\right|_{\symmetricpower[n](\mathcal{H})\otimes\symmetricpower[n](\mathcal{K})}\le A^{\frac{1-\alpha}{\alpha}}_{\mathcal{H},n}\otimes A^{\frac{1-\alpha}{\alpha}}_{\mathcal{K},n}$
\item\label[property]{it:observabledirectsum} $\left.A^{\frac{1-\alpha}{\alpha}}_{\mathcal{H}\oplus\mathcal{K},m+n}\right|_{\symmetricpower[m](\mathcal{H})\otimes\symmetricpower[n](\mathcal{K})}\le 2^{(m+n)\frac{1-\alpha}{\alpha}h(\frac{m}{m+n})}A^{\frac{1-\alpha}{\alpha}}_{\mathcal{H},m}\otimes A^{\frac{1-\alpha}{\alpha}}_{\mathcal{K},n}$.
\end{enumerate}
\begin{remark}\label{rem:observablecommutation}
\Cref{it:observableisometry} implies in particular that $A_{\mathcal{H},n}$ commutes with operators of the form $(S_1\otimes\dots\otimes S_k)^{\otimes n}$ where $S_j:\mathcal{H}_j\to\mathcal{H}_j$ are arbitrary linear maps. Indeed, when the $S_j$ are unitaries, this is a special case of \cref{it:observableisometry}, and it extends to the complexification since commutation is a (holomorphic) polynomial condition. The extension to noninvertible maps follows by continuity. More generally, if $B_j\in\boundeds(\mathcal{H}_j^{\otimes n})$ is a symmetrization of a tensor product of operators on $\mathcal{H}_j$ for all $j$, then $B_1\otimes\dots\otimes B_k$ also commutes with $A_{\mathcal{H},n}$, which can be seen by writing $B_j$ as a linear combination of tensor powers.
\end{remark}

\begin{proposition}
Let $0<\alpha\le\beta<1$ and suppose that the observables $A_{\mathcal{H},n}$ satisfy \cref{it:observableisometry,it:observablebound,it:observablesupermultiplicative,it:observabletensorproduct,it:observabledirectsum} with parameter $\alpha$. Then they also satisfy the properties with parameter $\beta$ instead of $\alpha$.
\end{proposition}
\begin{proof}
Recall that $x\mapsto x^t$ is an operator monotone and operator concave function on $[0,\infty)$ when $t\in(0,1]$. Note also that $\frac{1-\alpha}{\alpha}$ is a decreasing function of $\alpha$ when $\alpha>0$, therefore $t=\frac{1-\beta}{\beta}/\frac{1-\alpha}{\alpha}\in(0,1]$. This implies for example
\begin{equation}
\begin{split}
\left.A^{\frac{1-\beta}{\beta}}_{\mathcal{H}\otimes\mathcal{K},n}\right|_{\symmetricpower[n](\mathcal{H})\otimes\symmetricpower[n](\mathcal{K})}
 & = \left.\left(A^{\frac{1-\alpha}{\alpha}}_{\mathcal{H}\otimes\mathcal{K},n}\right)^t\right|_{\symmetricpower[n](\mathcal{H})\otimes\symmetricpower[n](\mathcal{K})}  \\
 & \le \left(\left.A^{\frac{1-\alpha}{\alpha}}_{\mathcal{H}\otimes\mathcal{K},n}\right|_{\symmetricpower[n](\mathcal{H})\otimes\symmetricpower[n](\mathcal{K})}\right)^t  \\
 & \le \left(A^{\frac{1-\alpha}{\alpha}}_{\mathcal{H},n}\otimes A^{\frac{1-\alpha}{\alpha}}_{\mathcal{K},n}\right)^t  \\
 & = A^{\frac{1-\beta}{\beta}}_{\mathcal{H},n}\otimes A^{\frac{1-\beta}{\beta}}_{\mathcal{K},n}
\end{split}
\end{equation}
The first inequality uses that if $f$ is operator concave and $K$ is a contraction then $f(K^*AK)\ge K^*f(A)K$ (see e.g. \cite[Theorem V.2.3]{bhatia2013matrix}), while the second inequality uses \cref{it:observabletensorproduct} (for $\alpha$) and operator monotonicity. This proves \cref{it:observabletensorproduct} for $\beta$. The remaining properties either do not involve the parameter at all or are proved in a similar way.
\end{proof}

Constructing families of observables satisfying \cref{it:observableisometry,it:observablebound,it:observablesupermultiplicative,it:observabletensorproduct,it:observabledirectsum} appears to be a difficult task. In the following proposition we introduce a particular family defined for order-2 tensors, and later we show how can such sequences be combined in a nontrivial way.
\begin{proposition}\label{prop:bipartiteobservables}
Let $k=2$. The observables
\begin{equation}\label{eq:bipartiteobservables}
A_{\mathcal{H},n}=\left.\sum_{\lambda\vdash n}2^{n\entropy(\lambda/n)}P^{\mathcal{H}_1}_\lambda\otimes P^{\mathcal{H}_2}_\lambda\right|_{\symmetricpower[n](\mathcal{H})}
\end{equation}
satisfy \cref{it:observableisometry,it:observablebound,it:observablesupermultiplicative,it:observabletensorproduct,it:observabledirectsum} with $c_\mathcal{H}=\min\{\dim\mathcal{H}_1,\dim\mathcal{H}_2\}$ and any $\alpha\in(0,1)$.
\end{proposition}
\begin{proof}
\cref{it:observableisometry} follows from the fact that the isotypic projections are invariant under tensor powers of isometries. The lower bound $A_{\mathcal{H},n}\ge I$ holds because the entropy is nonnegative and the symmetric subspace is contained in the range of the sum of the projections, a consequence of the isomorphism
\begin{equation}
\symmetricpower[n](\mathcal{H}_1\otimes\mathcal{H}_2)\simeq\bigoplus_{\lambda\vdash n}\mathbb{S}_\lambda(\mathcal{H}_1)\otimes\mathbb{S}_\lambda(\mathcal{H}_2).
\end{equation}
Since the terms with $l(\lambda)\ge\min\{\dim\mathcal{H}_1,\dim\mathcal{H}_2\}$ vanish and $\entropy(\lambda/n)\le\log l(\lambda)$, \cref{it:observablebound} is satisfied with $c_\mathcal{H}=\min\{\dim\mathcal{H}_1,\dim\mathcal{H}_2\}$.

We prove \cref{it:observablesupermultiplicative}. Let $m,n\in\naturals$, $\mu\vdash m$, $\nu\vdash n$, $\lambda\vdash m+n$. Then the projectons $P^{\mathcal{H}_1}_\mu\otimes P^{\mathcal{H}_1}_\nu$ and $P^{\mathcal{H}_1}_\lambda$ commute and their product is nonzero only if $c^\lambda_{\mu\nu}\neq 0$, which implies $m\entropy(\mu/m)+n\entropy(\nu/n)\le(m+n)\entropy(\lambda/(m+n))$, therefore
\begin{equation}
\begin{split}
\sum_{\lambda\vdash m+n}2^{(m+n)\frac{1-\alpha}{\alpha}\entropy(\lambda/(m+n))}P^{\mathcal{H}_1}_\lambda
 & = \sum_{\substack{\lambda\vdash m+n  \\  \mu\vdash m  \\  \nu\vdash n}}2^{(m+n)\frac{1-\alpha}{\alpha}\entropy(\lambda/(m+n))}P^{\mathcal{H}_1}_\lambda\left(P^{\mathcal{H}_1}_\mu\otimes P^{\mathcal{H}_1}_\nu\right)  \\
 & \ge \sum_{\substack{\lambda\vdash m+n  \\  \mu\vdash m  \\  \nu\vdash n}}2^{m\frac{1-\alpha}{\alpha}\entropy(\mu/m)+n\frac{1-\alpha}{\alpha}\entropy(\nu/n)}P^{\mathcal{H}_1}_\lambda\left(P^{\mathcal{H}_1}_\mu\otimes P^{\mathcal{H}_1}_\nu\right)  \\
 & = \left(\sum_{\mu\vdash m}2^{m\frac{1-\alpha}{\alpha}\entropy(\mu/m)}P^{\mathcal{H}_1}_\mu\right)\otimes\left(\sum_{\nu\vdash n}2^{n\frac{1-\alpha}{\alpha}\entropy(\mu/n)}P^{\mathcal{H}_1}_\nu\right).
\end{split}
\end{equation}
Take the tensor product of this inequality with the identity operator on $\mathcal{H}_2^{\otimes(m+n)}$, and then restrict to $\symmetricpower[m](\mathcal{H})\otimes\symmetricpower[n](\mathcal{H})$. Then the right hand side becomes $A^{\frac{1-\alpha}{\alpha}}_{\mathcal{H},m}\otimes A^{\frac{1-\alpha}{\alpha}}_{\mathcal{H},n}$. Finally, restrict to $\symmetricpower[m+n](\mathcal{H})$ to obtain \cref{it:observablesupermultiplicative}, using that $\symmetricpower[m+n](\mathcal{H})$ is a subspace of $\symmetricpower[m](\mathcal{H})\otimes\symmetricpower[n](\mathcal{H})$.

We prove \cref{it:observabletensorproduct}. Let $\mathcal{H}=\mathcal{H}_1\otimes\mathcal{H}_2$ and $\mathcal{K}=\mathcal{K}_1\otimes\mathcal{K}_2$, $n\in\naturals$, $\lambda,\mu,\nu\vdash n$. The projections $P^{\mathcal{H}_1\otimes\mathcal{K}_1}_\lambda$ and $P^{\mathcal{H}_1}_\mu\otimes P^{\mathcal{K}_1}_\nu$ commute and their product is nonzero only if $g_{\lambda\mu\nu}\neq 0$, which implies $\entropy(\lambda/n)\le\entropy(\mu/n)+\entropy(\nu/n)$, therefore
\begin{equation}
\begin{split}
\sum_{\lambda\vdash n}2^{n\frac{1-\alpha}{\alpha}\entropy(\lambda/n)}P^{\mathcal{H}_1\otimes\mathcal{K}_1}_\lambda
 & = \sum_{\lambda,\mu,\nu\vdash n}2^{n\frac{1-\alpha}{\alpha}\entropy(\lambda/n)}P^{\mathcal{H}_1\otimes\mathcal{K}_1}_\lambda\left(P^{\mathcal{H}_1}_\mu\otimes P^{\mathcal{K}_1}_\nu\right)  \\
 & \le \sum_{\lambda,\mu,\nu\vdash n}2^{n\frac{1-\alpha}{\alpha}\entropy(\mu/n)+n\frac{1-\alpha}{\alpha}\entropy(\nu/n)}P^{\mathcal{H}_1\otimes\mathcal{K}_1}_\lambda\left(P^{\mathcal{H}_1}_\mu\otimes P^{\mathcal{K}_1}_\nu\right)  \\
 & = \left(\sum_{\mu\vdash n}2^{n\frac{1-\alpha}{\alpha}\entropy(\mu/n)}P^{\mathcal{H}_1}_\mu\right)\otimes\left(\sum_{\nu\vdash n}2^{n\frac{1-\alpha}{\alpha}\entropy(\nu/n)}P^{\mathcal{H}_1}_\nu\right).
\end{split}
\end{equation}
Take the tensor product of this inequality with the identity operator on $(\mathcal{H}_2\otimes\mathcal{K}_2)^{\otimes n}$, and then restrict to $\symmetricpower[n](\mathcal{H}\otimes\mathcal{K})$. Then the right hand side becomes $A^{\frac{1-\alpha}{\alpha}}_{\mathcal{H},n}\otimes A^{\frac{1-\alpha}{\alpha}}_{\mathcal{K},n}$. Finally, restrict to $\symmetricpower[n](\mathcal{H})\otimes\symmetricpower[n](\mathcal{K})$ to obtain \cref{it:observabletensorproduct}.

We prove \cref{it:observabledirectsum}. Let $m,n\in\naturals$, $\lambda\vdash m+n$, $\mu\vdash m$, $\nu\vdash n$. The projections $P^{\mathcal{H}_1\oplus\mathcal{K}_1}_\lambda$ and $I_{\binom{m}{m+n}}\otimes P^{\mathcal{H}_1}_\mu\otimes P^{\mathcal{K}_1}_\nu$ commute and their product is nonzero only if $c^\lambda_{\mu\nu}\neq 0$, which implies $(m+n)\entropy(\lambda/(m+n))\le m\entropy(\mu/m)+n\entropy(\nu/n)+(m+n)h(m/(m+n))$, therefore
\begin{multline}
\left.\sum_{\lambda\vdash m+n}2^{(m+n)\frac{1-\alpha}{\alpha}\entropy(\lambda/(m+n))}P^{\mathcal{H}_1\oplus\mathcal{K}_1}_\lambda\right|_{\complexes^{\binom{m}{m+n}}\otimes\mathcal{H}_1^{\otimes m}\otimes\mathcal{K}_1^{\otimes n}}  \\
  = \sum_{\substack{\lambda\vdash m+n  \\  \mu\vdash m  \\  \nu\vdash n}}2^{(m+n)\frac{1-\alpha}{\alpha}\entropy(\lambda/(m+n))}P^{\mathcal{H}_1\oplus\mathcal{K}_1}_\lambda\left(I_{\binom{m}{m+n}}\otimes P^{\mathcal{H}_1}_\mu\otimes P^{\mathcal{K}_1}_\nu\right)  \\
  \le \sum_{\substack{\lambda\vdash m+n  \\  \mu\vdash m  \\  \nu\vdash n}}2^{m\frac{1-\alpha}{\alpha}\entropy(\mu/m)+n\frac{1-\alpha}{\alpha}\entropy(\nu/n)+(m+n)\frac{1-\alpha}{\alpha}h(\frac{m}{m+n})}P^{\mathcal{H}_1\oplus\mathcal{K}_1}_\lambda\left(I_{\binom{m}{m+n}}\otimes P^{\mathcal{H}_1}_\mu\otimes P^{\mathcal{K}_1}_\nu\right)  \\
  = 2^{(m+n)\frac{1-\alpha}{\alpha}h(\frac{m}{m+n})}I_{\binom{m}{m+n}}\otimes\left(\sum_{\mu\vdash m}2^{m\frac{1-\alpha}{\alpha}\entropy(\mu/m)}P^{\mathcal{H}_1}_\mu\right)\otimes\left(\sum_{\nu\vdash n}2^{n\frac{1-\alpha}{\alpha}\entropy(\nu/n)}  P^{\mathcal{K}_1}_\nu\right).
\end{multline}
Take the tensor product of this inequality with $\complexes^{\binom{m}{m+n}}\otimes\mathcal{H}_2^{\otimes m}\otimes\mathcal{K}_2^{\otimes n}$ (considered as a subspace of $(\mathcal{H}_2\oplus\mathcal{K}_2)^{\otimes(m+n)}$), and then restrict to $\symmetricpower[m+n](\mathcal{H}\oplus\mathcal{K})$ to get \cref{it:observabledirectsum}.
\end{proof}

\begin{theorem}\label{thm:upper}
Let $A$ be a family of observables satisfying \cref{it:observableisometry,it:observablebound,it:observablesupermultiplicative,it:observabletensorproduct,it:observabledirectsum} with parameter $\alpha$. The limit
\begin{equation}\label{eq:logupperdef}
E^{\alpha,A}(\psi)=\lim_{n\to\infty}\frac{1}{n}\frac{\alpha}{1-\alpha}\log\bra{\psi}^{\otimes n}A^{\frac{1-\alpha}{\alpha}}_{\mathcal{H},n}\ket{\psi}^{\otimes n}=-\lim_{n\to\infty}\frac{1}{n} \sandwiched[\alpha]{\vectorstate{\psi}^{\otimes n}}{A_{\mathcal{H},n}}
\end{equation}
(see \eqref{eq:reformulatedivergencetfamily}) exists for all $\psi\in\mathcal{H}\setminus\{0\}$, and
\begin{equation}\label{eq:upperdef}
F^{\alpha,A}(\psi)=\begin{cases}
2^{(1-\alpha)E^{\alpha,A}(\psi)} & \text{if $\psi\neq 0$}  \\
0 & \text{otherwise}
\end{cases}
\end{equation}
is an upper functional of order $\alpha$. Moreover, $F^{\alpha,A}$ can be bounded as
\begin{equation}\label{eq:upperbounds}
\norm{\psi}^{2\alpha}\le F^{\alpha,A}(\psi)\le\norm{\psi}^{2\alpha}c_\mathcal{H}^{1-\alpha}
\end{equation}
when $\psi\in\mathcal{H}$.
\end{theorem}
Bringing out the $1-\alpha$ factor may seem arbitrary, but we will benefit from this choice of definition.

\begin{proof}
Let $\psi\in\mathcal{H}\setminus\{0\}$. For all $m,n\in\naturals$ we have
\begin{equation}
\begin{split}
\log\bra{\psi}^{\otimes (n+m)}A^{\frac{1-\alpha}{\alpha}}_{\mathcal{H},n+m}\ket{\psi}^{\otimes (n+m)}
 & \ge 
 \log\bra{\psi}^{\otimes (n+m)}A^{\frac{1-\alpha}{\alpha}}_{\mathcal{H},n}\otimes A^{\frac{1-\alpha}{\alpha}}_{\mathcal{H},m}\ket{\psi}^{\otimes (n+m)}
 \\
 & =  \log\bra{\psi}^{\otimes n}A^{\frac{1-\alpha}{\alpha}}_{\mathcal{H},n}\ket{\psi}^{\otimes n}+
  \log\bra{\psi}^{\otimes m} A^{\frac{1-\alpha}{\alpha}}_{\mathcal{H},m}\ket{\psi}^{\otimes m}
\end{split}
\end{equation}
 \Cref{it:observablebound} and \cref{it:observablesupermultiplicative,} imply that
\begin{equation}
\begin{split}
\log\bra{\psi}^{\otimes n}A^{\frac{1-\alpha}{\alpha}}_{\mathcal{H},n}\ket{\psi}^{\otimes n}
 & \le \log\bra{\psi}^{\otimes n}c_\mathcal{H}^{{\frac{1-\alpha}{\alpha}}n}I\ket{\psi}^{\otimes n}  \\
 & = n\log\norm{\psi}^2+{{\frac{1-\alpha}{\alpha}}n}\log c_\mathcal{H}
\end{split}
\end{equation}
and
\begin{equation}
\begin{split}
\log\bra{\psi}^{\otimes n}A^{\frac{1-\alpha}{\alpha}}_{\mathcal{H},n}\ket{\psi}^{\otimes n}
 & \ge \log\bra{\psi}^{\otimes n}I\ket{\psi}^{\otimes n}  \\
 & = n\log\norm{\psi}^2.
\end{split}
\end{equation}
This means that the sequence $n\mapsto\log\bra{\psi}^{\otimes n}A^{\frac{1-\alpha}{\alpha}}_{\mathcal{H},n}\ket{\psi}^{\otimes n}$ is subadditive and has linear upper and lower bounds, therefore, by the Fekete lemma the limit in \eqref{eq:logupperdef} exists and $\frac{\alpha}{1-\alpha}\log\norm{\psi}^2\le E^{\alpha,A}(\psi)\le \frac{\alpha}{1-\alpha}\log\norm{\psi}^2+\frac{\alpha}{1-\alpha}\log c_\mathcal{H}$. Note also that for every $n\in\naturals$ we have
\begin{equation}\label{eq:logupperlowerbound}
E^{\alpha,A}(\psi)\ge\frac{1}{n}\frac{\alpha}{1-\alpha}\log\bra{\psi}^{\otimes n}A^{\frac{1-\alpha}{\alpha}}_{\mathcal{H},n}\ket{\psi}^{\otimes n}.
\end{equation}

If $p>0$ and $\psi\in\mathcal{H}\setminus\{0\}$, then
\begin{equation}
\begin{split}
E^{\alpha,A}(\sqrt{p}\psi)
 & =\lim_{n\to\infty} \frac{1}{n}\frac{\alpha}{1-\alpha}\log\bra{\psi}^{\otimes n}A^{\frac{1-\alpha}{\alpha}}_{\mathcal{H},n}\ket{\psi}^{\otimes n} +\frac{\alpha}{1-\alpha}\log p \\
 & = E^{\alpha,A}(\psi)+\frac{\alpha}{1-\alpha}\log p,
\end{split}
\end{equation}
i.e. $F^{\alpha,A}(\sqrt{p}\psi)=p^\alpha F^{\alpha,A}(\psi)$.

Let $\psi\in\mathcal{H}\setminus\{0\}$ and $\varphi\in\mathcal{K}\setminus\{0\}$. Then $(\psi\otimes\varphi)^{\otimes n}\in\symmetricpower[n](\mathcal{H})\otimes\symmetricpower[n](\mathcal{K})$, therefore by \cref{it:observabletensorproduct,} we have
\begin{equation}
\begin{split}
E^{\alpha,A}(\psi\otimes\varphi)
 & = \lim_{n\to\infty} \frac{1}{n}\frac{\alpha}{1-\alpha}\log\bra{\psi\otimes\varphi}^{\otimes n}A^{\frac{1-\alpha}{\alpha}}_{\mathcal{H}\otimes\mathcal{K},n}\ket{\psi\otimes\varphi}^{\otimes n}\\
& \le \lim_{n\to\infty} \frac{1}{n}\frac{\alpha}{1-\alpha}\log\bra{\psi\otimes\varphi}^{\otimes n}A^{\frac{1-\alpha}{\alpha}}_{\mathcal{H},n}\otimes A^{\frac{1-\alpha}{\alpha}}_{\mathcal{K},n}\ket{\psi\otimes\varphi}^{\otimes n}\\
 & = E^{\alpha,A}(\psi)+E^{\alpha,A}(\varphi),
\end{split}
\end{equation}
i.e. $F^{\alpha,A}$ is submultiplicative.

Also, by \cref{it:observabledirectsum}, \eqref{eq:logupperlowerbound} and using the short-hand notations $W_m^*$ for the isometry in \eqref{eq:directsumpowerrestriction}, for every $q\in(0,1)$ we have
\begin{equation}
\begin{split}
E^{\alpha,A}(\psi\oplus\varphi)
 & = \lim_{n\to\infty} \frac{1}{n}\frac{\alpha}{1-\alpha}\log\bra{\psi\oplus\varphi}^{\otimes n}A^{\frac{1-\alpha}{\alpha}}_{\mathcal{H}\oplus\mathcal{K},n}\ket{\psi\oplus\varphi}^{\otimes n}  \\
 & \le \lim_{n\to\infty} \frac{1}{n}\frac{\alpha}{1-\alpha}\log\bra{\psi\oplus\varphi}^{\otimes n}(n+1)\sum_{m=0}^n 
 W_mW_m^* A^{\frac{1-\alpha}{\alpha}}_{\mathcal{H}\oplus\mathcal{K},n}W_mW_m^*\ket{\psi\oplus\varphi}^{\otimes n}  \\
 & \le \lim_{n\to\infty} \max_{m\in \{0,\dots,n \}} \frac{1}{n} \Big[  \frac{\alpha}{1-\alpha}\log\bra{\psi\oplus\varphi}^{\otimes n}
 W_mW_m^*
 A^{\frac{1-\alpha}{\alpha}}_{\mathcal{H}\oplus\mathcal{K},n} W_mW_m^*
 \ket{\psi\oplus\varphi}^{\otimes n}\\
&\qquad{}+\frac{\alpha}{1-\alpha}\log(n+1)^2 \Big]  \\
& = \lim_{n\to\infty} \max_{m\in \{0,\dots,n \}} \frac{1}{n} \Big[  \frac{\alpha}{1-\alpha}\log{\binom{n}{m}}\bra{\psi}^{\otimes m}\bra{\varphi}^{\otimes n-m}W_m^*A^{\frac{1-\alpha}{\alpha}}_{\mathcal{H}\oplus\mathcal{K},n}W_m
\ket{\psi}^{\otimes m}\ket{\varphi}^{\otimes n-m} \Big]  \\
& \le \lim_{n\to\infty} \max_{m\in \{0,\dots,n \}} \frac{1}{n} \Big[  \frac{\alpha}{1-\alpha}\log{\binom{n}{m}}\bra{\psi}^{\otimes m}\bra{\varphi}^{\otimes n-m}
2^{n\frac{1-\alpha}{\alpha}h(\frac{m}{n})}
\\ &\qquad{}\cdot
 A^{\frac{1-\alpha}{\alpha}}_{\mathcal{H},m}\otimes A^{\frac{1-\alpha}{\alpha}}_{\mathcal{K},n-m}
\ket{\psi}^{\otimes m}\ket{\varphi}^{\otimes n-m}\Big]  \\
 & = \lim_{n\to\infty}\max_{m\in\{0,\ldots,n\}}\frac{1}{n}\Big[\frac{\alpha}{1-\alpha}\log\bra{\psi}^{\otimes m}
A^{\frac{1-\alpha}{\alpha}}_{\mathcal{H},m}\ket{\psi}^{\otimes m}+\log \bra{\varphi}^{\otimes n-m}A^{\frac{1-\alpha}{\alpha}}_{\mathcal{K},n-m}
\ket{\varphi}^{\otimes n-m}  \\
 &\qquad{}+\frac{\alpha}{1-\alpha}\log\binom{n}{m}+nh\left(\frac{m}{n}\right)\Big]  \\
 & \le \lim_{n\to\infty}\max_{m\in\{0,\ldots,n\}}\left[\frac{m}{n} E^{\alpha,A}(\psi)+\frac{n-m}{n}E^{\alpha,A}(\varphi)+\frac{1}{1-\alpha}h\left(\frac{m}{n}\right)\right]  \\
 & \le \max_{q\in[0,1]}\left[q E^{\alpha,A}(\psi)+(1-q)E^{\alpha,A}(\varphi)+\frac{1}{1-\alpha}h(q)\right]  \\
 & = \frac{1}{1-\alpha}\log 2^{(1-\alpha)E^{\alpha,A}(\psi)}+2^{(1-\alpha)E^{\alpha,A}(\varphi)},
\end{split}
\end{equation}
where the first inequality is the pinching inequality. The last inequality is \cite[Eq. (2.13)]{strassen1991degeneration}. This means that $F^{\alpha,A}$ is subadditive.

Let $\psi\in\mathcal{H}$ and $\Pi\in\boundeds(\mathcal{H}_j)$ be a projection, which we identify with an element of $\boundeds(\mathcal{H})$ by considering the tensor product with the identity operators on the remaining factors. Let $\psi_1=\Pi\psi$ and $\psi_2=(I-\Pi)\psi$. For $m,n\in\naturals$ such that $m\le n$ we form the projection
\begin{equation}
P_{n,m}=\frac{1}{m!(n-m)!}\sum_{\sigma\in S_n}\sigma\cdot(\Pi^{\otimes m}\otimes(I-\Pi)^{\otimes n-m}).
\end{equation}
These operators commute with $A_{\mathcal{H},n}$ by \cref{it:observableisometry} (see \cref{rem:observablecommutation}).

Let $q\in(0,1)$ and choose $m=\lfloor nq\rfloor$ with $n$ so large that $m\neq 0$ and $m\neq n$. By \eqref{eq:logupperlowerbound} we have
\begin{equation}
\begin{split}
E^{\alpha,A}(\psi)
 & \ge \frac{1}{n}\frac{\alpha}{1-\alpha}\log\bra{\psi}^{\otimes n}A^{\frac{1-\alpha}{\alpha}}_{\mathcal{H},n}\ket{\psi}^{\otimes n} \\
 & \ge\frac{1}{n}\frac{\alpha}{1-\alpha}\log\bra{\psi}^{\otimes n}P_{n,m}A^{\frac{1-\alpha}{\alpha}}_{\mathcal{H},n}\ket{\psi}^{\otimes n} \\
  & =\frac{1}{n}\frac{\alpha}{1-\alpha}\log\bra{\psi}^{\otimes n}\frac{1}{\lfloor qn\rfloor!(n-\lfloor qn\rfloor)!}\sum_{\sigma\in S_n}\sigma\cdot(\Pi^{\otimes \lfloor qn\rfloor}\otimes(I-\Pi)^{\otimes n-\lfloor qn\rfloor})
  A^{\frac{1-\alpha}{\alpha}}_{\mathcal{H},n}\ket{\psi}^{\otimes n} \\
  & \ge\frac{1}{n}\frac{\alpha}{1-\alpha}\log {\binom{n}{\lfloor qn\rfloor}} \bra{\psi}^{\otimes n}
  (\Pi^{\otimes \lfloor qn\rfloor}\otimes(I-\Pi)^{\otimes n-\lfloor qn\rfloor})
  A^{\frac{1-\alpha}{\alpha}}_{\mathcal{H},\lfloor qn\rfloor}\otimes A^{\frac{1-\alpha}{\alpha}}_{\mathcal{H},n-\lfloor qn\rfloor}\Big|_{\symmetricpower[n](\mathcal{H})}
  \ket{\psi}^{\otimes n} \\
  &= \frac{1}{n}\frac{\alpha}{1-\alpha}\log{\binom{n}{\lfloor qn\rfloor}}\bra{\psi_1}^{\otimes{\lfloor qn\rfloor}}
  A^{\frac{1-\alpha}{\alpha}}_{\mathcal{H},\lfloor qn\rfloor} \ket{\psi_1}^{\otimes {\lfloor qn\rfloor}} \\
&\qquad\qquad\qquad\qquad\quad  \cdot \bra{\psi_2}^{\otimes n- {\lfloor qn\rfloor}} 
  A^{\frac{1-\alpha}{\alpha}}_{\mathcal{H},{n-\lfloor qn\rfloor}}
  \ket{\psi_2}^{\otimes n- {\lfloor qn\rfloor}} \\
 & = \frac{\lfloor qn\rfloor}{n}\frac{1}{\lfloor qn\rfloor}\frac{\alpha}{1-\alpha}\log\bra{\psi_1}^{\otimes{\lfloor qn\rfloor}}
  A^{\frac{1-\alpha}{\alpha}}_{\mathcal{H},\lfloor qn\rfloor} \ket{\psi_1}^{\otimes {\lfloor qn\rfloor}}  \\
 &\qquad{}+\frac{n-\lfloor qn\rfloor}{n}\frac{1}{n-\lfloor qn\rfloor}\frac{\alpha}{1-\alpha}
 \log
 \bra{\psi_2}^{\otimes n- {\lfloor qn\rfloor}}
 A^{\frac{1-\alpha}{\alpha}}_{\mathcal{H},{n-\lfloor qn\rfloor}}
  \ket{\psi_2}^{\otimes n- {\lfloor qn\rfloor}} 
  \\
 &\qquad{}
 +\frac{\alpha}{1-\alpha}\frac{1}{n}\log\binom{n}{\lfloor qn\rfloor}.
\end{split}
\end{equation}
In the first equality we used the $\symmetricpower[n]$ symmetry of the observable.
Taking the limit $n\to\infty$ gives
\begin{equation}
E^{\alpha,A}(\psi)\ge qE^{\alpha,A}(\psi_1)+(1-q)E^{\alpha,A}(\psi_2)+\frac{\alpha}{1-\alpha}h(q).
\end{equation}
This is true for all $q\in(0,1)$, therefore we can maximize the right hand side over $q$, which leads to the inequality
\begin{equation}
F^{\alpha,A}(\psi)^{1/\alpha} \ge F^{\alpha,A}(\Pi\psi)^{1/\alpha}+F^{\alpha,A}((I-\Pi)\psi)^{1/\alpha}.
\end{equation}

\end{proof}

\begin{lemma}\label{lem:bipartiteasymptotic}
Let $k=2$ and $A_{\mathcal{H},n}$ be as in \cref{prop:bipartiteobservables} and $\psi\in\mathcal{H}_1\otimes\mathcal{H}_2$ a unit vector.
\begin{enumerate}
\item If $\alpha\in(0,1)\cup(1,\infty]$, then
\begin{equation}
\lim_{n\to\infty}\frac{1}{n}\frac{\alpha}{1-\alpha}\log\bra{\psi^{\otimes n}}A_{\mathcal{H},n}^{\frac{1-\alpha}{\alpha}}\ket{\psi^{\otimes n}}=\entropy_{\alpha}(\Tr_2\vectorstate{\psi}).
\end{equation}
\item $\lim_{\alpha\to 1}\frac{\alpha}{1-\alpha}\log\bra{\psi^{\otimes n}}A_{\mathcal{H},n}^{\frac{1-\alpha}{\alpha}}\ket{\psi^{\otimes n}}=\bra{\psi^{\otimes n}}\log A_{\mathcal{H},n}\ket{\psi^{\otimes n}}$ and
\begin{equation}
\lim_{n\to\infty}\frac{1}{n}\bra{\psi^{\otimes n}}\log A_{\mathcal{H},n}\ket{\psi^{\otimes n}}=\entropy(\Tr_2\vectorstate{\psi}).
\end{equation}
\item
In particular by the scaling property of the upper functional, for an arbitrary vector $\psi\in\mathcal{H}_1\otimes\mathcal{H}_2\setminus\{0\}$ we have
\begin{equation}
E^{\alpha,A}(\psi)=\frac{\alpha}{1-\alpha}\log\norm{\psi}^2+\entropy_{\alpha}(\Tr_2\vectorstate{\psi}/\norm{\psi}^2)
\end{equation}
and
\begin{equation}\label{eq:bipartitespectrum}
F^{\alpha,A}(\psi)=\Tr(\Tr_2\vectorstate{\psi})^\alpha.
\end{equation}
For $\alpha\in[0,1]$, these are exactly the elements of the asymptotic spectrum of LOCC transformations of bipartite states \cite{jensen2019asymptotic}.
\end{enumerate}
\end{lemma}
\begin{proof}
Let $r=(r_1,\dots,r_d)$ ($d=\min\{\dim\mathcal{H}_1,\dim\mathcal{H}_2\}$) be the Schmidt coefficients of $\psi$ with multiplicities and listed in decreasing order. By \eqref{eq:quantumtypeestimate} we have
\begin{equation}
(n+d)^{-d(d+1)/2}2^{-n\relativeentropy{\frac{\lambda}{n}}{r}}\le\bra{\psi^{\otimes n}}(P^{\mathcal{H}_1}_\lambda\otimes P^{\mathcal{H}_2}_\lambda)\ket{\psi^{\otimes n}}\le(n+1)^{d(d-1)/2}2^{-n\relativeentropy{\frac{\lambda}{n}}{r}},
\end{equation}
which implies (using that the number of partitions of $n$ into at most $d$ parts is bounded by $(n+1)^d$)
\begin{equation}
\begin{split}
\bra{\psi^{\otimes n}}A_{\mathcal{H},n}^{\frac{1-\alpha}{\alpha}}\ket{\psi^{\otimes n}}
 & = \sum_{\lambda\vdash n}2^{\frac{1-\alpha}{\alpha}n\entropy(\lambda/n)}\bra{\psi^{\otimes n}}(P^{\mathcal{H}_1}_\lambda\otimes P^{\mathcal{H}_2}_\lambda)\ket{\psi^{\otimes n}}  \\
 & \le \sum_{\lambda\vdash n}2^{\frac{1-\alpha}{\alpha}n\entropy(\lambda/n)-n\relativeentropy{\frac{\lambda}{n}}{r}}(n+1)^{d(d-1)/2}  \\
 & \le (n+1)^{d(d+1)/2}\max_{\lambda\vdash n}2^{-n\left(\frac{\alpha-1}{\alpha}\entropy(\lambda/n)+\relativeentropy{\frac{\lambda}{n}}{r}\right)},
\end{split}
\end{equation}
and
\begin{equation}
\bra{\psi^{\otimes n}}A_{\mathcal{H},n}^{\frac{1-\alpha}{\alpha}}\ket{\psi^{\otimes n}}
 \ge (n+d)^{-d(d+1)/2}\max_{\lambda\vdash n}2^{-n\left(\frac{\alpha-1}{\alpha}\entropy(\lambda/n)+\relativeentropy{\frac{\lambda}{n}}{r}\right)}.
\end{equation}
It follows that
\begin{equation}
\begin{split}
-\lim_{n\to\infty}\frac{1}{n}\frac{\alpha}{\alpha-1}\log\bra{\psi^{\otimes n}}A_{\mathcal{H},n}^{\frac{1-\alpha}{\alpha}}\ket{\psi^{\otimes n}}
 & = \lim_{n\to\infty}\frac{\alpha}{\alpha-1}\min_{\lambda\vdash n}\left[\frac{\alpha-1}{\alpha}\entropy(\lambda/n)+\relativeentropy{\frac{\lambda}{n}}{r}\right]  \\
 & = \frac{\alpha}{\alpha-1}\min_{Q\in\distributions([d])}\left[\frac{\alpha-1}{\alpha}\entropy(Q)+\relativeentropy{Q}{r}\right]  \\
 & = \entropy_{\alpha}(\Tr_2\vectorstate{\psi}),
\end{split}
\end{equation}
in the last line using \eqref{eq:variationalRenyientropy}.

In the $\alpha\to 1$ case we obtain
\begin{equation}
\begin{split}
\lim_{n\to\infty}\frac{1}{n}\bra{\psi^{\otimes n}}\log A_{\mathcal{H},n}\ket{\psi^{\otimes n}}
 & = \lim_{n\to\infty}\sum_{\lambda\vdash n}\entropy\left(\frac{\lambda}{n}\right)\bra{\psi^{\otimes n}}(P^{\mathcal{H}_1}_\lambda\otimes P^{\mathcal{H}_2}_\lambda)\ket{\psi^{\otimes n}}  \\
 & = \entropy(\Tr_2\vectorstate{\psi}),
\end{split}
\end{equation}
using the continuity of the Shannon entropy and that the probabilities $\bra{\psi^{\otimes n}}(P^{\mathcal{H}_1}_\lambda\otimes P^{\mathcal{H}_2}_\lambda)\ket{\psi^{\otimes n}}$ are concentrated around $\lambda\approx nr$ \cite{alicki1988symmetry,keyl2001estimating,hayashi2002quantum,christandl2006structure}.
\end{proof}

\Cref{lem:bipartiteasymptotic} shows that the observables defined in \cref{prop:bipartiteobservables} correspond to the bipartite monotone homomorphisms. The following propositions provide tools to generate more families of observables from the bipartite ones. The idea behind \cref{prop:groupedobservables} is very simple: grouping the tensor factors turns upper functionals on order-$k'$ tensors to upper functionals on order-$k$ tensors for any $k>k'$, and the observation is that this can be done at the level of families of observables. The more interesting construction is that of \cref{prop:geometricobservables}, which allows one to combine families of observables using geometric means.
\begin{proposition}\label{prop:groupedobservables}
Let $A_{\mathcal{H},n}$ be observables defined for order-$k$ tensors that satisfy \cref{it:observableisometry,it:observablebound,it:observablesupermultiplicative,it:observabletensorproduct,it:observabledirectsum}. Let $k'\ge k$ and $f:[k']\to[k]$ a surjective map. For all $\mathcal{H}=\mathcal{H}_1\otimes\cdots\otimes\mathcal{H}_{k'}$, consider the isomorphism
\begin{equation}
V_\mathcal{H}:\mathcal{H}\to\bigotimes_{j=1}^k\left(\bigotimes_{j'\in f^{-1}(j)}\mathcal{H}_{j'}\right)
\end{equation}
that rearranges and groups the factors according to the function $f$. Then the observables
\begin{equation}
A'_{\mathcal{H},n}:=(V_\mathcal{H}^*)^{\otimes n}A_{V_\mathcal{H}\mathcal{H},n}V_\mathcal{H}^{\otimes n}
\end{equation}
on order-$k'$ tensors also satisfy \cref{it:observableisometry,it:observablebound,it:observablesupermultiplicative,it:observabletensorproduct,it:observabledirectsum} with the same constant $c_\mathcal{H}=c_{V_\mathcal{H}\mathcal{H}}$.
\end{proposition}
\begin{proof}
\Cref{it:observableisometry} is true for the new family because permuting the factors and grouping turns a product isometry to a product isometry with respect to the new factorization. Since $V_\mathcal{H}$ is an isometry, we have $(V_\mathcal{H}^*)^{\otimes n}I(V_\mathcal{H}^{\otimes n})=I$, therefore \cref{it:observablebound} is preserved. \Cref{it:observablesupermultiplicative,it:observabletensorproduct} are not sensitive to the product structure of the spaces $\mathcal{H},\mathcal{K}$, and $V_\mathcal{H}^{\otimes n}$ is an isometry between the respective symmetric tensor powers, therefore these inequalities are inherited as well.

In the case of (tensor) direct sums, $V_{\mathcal{H}\oplus\mathcal{K}}(\mathcal{H}\oplus\mathcal{K})$ is in general \emph{not} isomorphic to $(V_\mathcal{H}\mathcal{H})\oplus(V_\mathcal{K}\mathcal{K})$, but the latter is naturally a subspace of the former. However, \cref{it:observabledirectsum} only sees restrictions to tensor powers of an even smaller subspace, the vector space direct sum of $\mathcal{H}$ and $\mathcal{K}$ (respectively $V_\mathcal{H}\mathcal{H}$ and $V_\mathcal{K}\mathcal{K}$), and is therefore not affected by the grouping.
\end{proof}

\begin{proposition}\label{prop:geometricobservables}
Let $\geometricmean$ be an $r$-variable operator geometric mean, and let $A^{ (1)}_{\mathcal{H},n},\dots,A^{ (r)}_{\mathcal{H},n}$ be families satisfying \cref{it:observableisometry,it:observablebound,it:observablesupermultiplicative,it:observabletensorproduct,it:observabledirectsum} with constants $c^{(1)}_{\mathcal{H}},\dots,c^{(r)}_{\mathcal{H}}$ in \cref{it:observablebound}. Then the family $A_{\mathcal{H},n}$, where
\begin{equation}
A_{\mathcal{H},n}:=\geometricmean(A^{(1){\frac{1-\alpha}{\alpha}} }_{\mathcal{H},n},\dots,A^{(r){\frac{1-\alpha}{\alpha}} }_{\mathcal{H},n})^{\frac{\alpha}{1-\alpha}}
\end{equation}
also satisfies \cref{it:observableisometry,it:observablebound,it:observablesupermultiplicative,it:observabletensorproduct,it:observabledirectsum} with $c_\mathcal{H}=\geometricmean(c^{(1)}_{\mathcal{H}},\dots,c^{(r)}_{\mathcal{H}})$.
\end{proposition}
\begin{proof}
Let $U_j:\mathcal{H}_j\to\mathcal{K}_j$ be isometries and introduce the abbreviation $V=(U_1\otimes\dots\otimes U_k)^{\otimes n}$.
\Cref{it:observableisometry} implies that $A^{(i){\frac{1-\alpha}{\alpha}}}_{\mathcal{H},n}$ commutes with $VV^*$ (see \cref{rem:observablecommutation}), therefore
\begin{equation}
\begin{split}
(U_1^*\otimes\dots\otimes U_k^*)^{\otimes n}A^{\frac{1-\alpha}{\alpha}}_{\mathcal{K},n}(U_1\otimes\dots\otimes U_k)^{\otimes n}
 & = V^*\geometricmean(A^{(1){\frac{1-\alpha}{\alpha}}}_{\mathcal{K},n},\dots,A^{(r){\frac{1-\alpha}{\alpha}}}_{\mathcal{K},n})V  \\
 & = \geometricmean(V^*A^{(1){\frac{1-\alpha}{\alpha}}}_{\mathcal{K},n}V,\dots,V^*A^{(r){\frac{1-\alpha}{\alpha}}}_{\mathcal{K},n}V)  \\
 & = \geometricmean(A^{(1){\frac{1-\alpha}{\alpha}}}_{\mathcal{H},n},\dots,A^{(r){\frac{1-\alpha}{\alpha}}}_{\mathcal{H},n})  \\
 & = A^{\frac{1-\alpha}{\alpha}}_{\mathcal{H},n}
\end{split}
\end{equation}
by \cref{it:gmeandirectsum,it:gmeanunitary}. This means that $A_{\mathcal{H},n}$ satisfies \cref{it:observableisometry}.

By \cref{it:gmeanmonotone} and \cref{it:observablebound} for $A^{(i)}_{\mathcal{H},n}$ we have
\begin{equation}
\begin{split}
I
 & = \geometricmean(I,\dots,I)  \\
 & \le \underbrace{\geometricmean(A^{ (1){\frac{1-\alpha}{\alpha}}}_{\mathcal{H},n},\dots,A^{(r){\frac{1-\alpha}{\alpha}} }_{\mathcal{H},n})}_{A^{\frac{1-\alpha}{\alpha}}_{\mathcal{H},n}}  \\
 & \le \geometricmean((c^{(1){\frac{1-\alpha}{\alpha}}}_\mathcal{H})^nI,\dots,(c^{(r){\frac{1-\alpha}{\alpha}}}_\mathcal{H})^nI)  \\
 & = \geometricmean(c^{(1)}_{\mathcal{H}},\dots,c^{(r)}_{\mathcal{H}})^{n{\frac{1-\alpha}{\alpha}}}I,
\end{split}
\end{equation}
which verifies \cref{it:observablebound} for $A_{\mathcal{H},n}$ with $c_\mathcal{H}=\geometricmean(c^{(1)}_{\mathcal{H}},\dots,c^{(r)}_{\mathcal{H}})$.

To see that \cref{it:observablesupermultiplicative} is satisfied by $A_{\mathcal{H},n}$, we use \cref{it:observablesupermultiplicative} for $A^{ (i)}_{\mathcal{H},n}$ together with \cref{it:gmeantensorproduct,rem:gmeanproperties,it:gmeanmonotone}:
\begin{equation}
\begin{split}
\left.(A^{\frac{1-\alpha}{\alpha}}_{\mathcal{H},m}\otimes A^{\frac{1-\alpha}{\alpha}}_{\mathcal{H},n})\right|_{\symmetricpower[m+n](\mathcal{H})}
 & = \left.(\geometricmean(A^{(1){\frac{1-\alpha}{\alpha}}}_{\mathcal{H},m},\dots,A^{ (r){\frac{1-\alpha}{\alpha}}}_{\mathcal{H},m})\otimes \geometricmean(A^{ (1){\frac{1-\alpha}{\alpha}}}_{\mathcal{H},n},\dots,A^{ (r){\frac{1-\alpha}{\alpha}}}_{\mathcal{H},n}))\right|_{\symmetricpower[m+n](\mathcal{H})}  \\
 & = \left.(\geometricmean(A^{ (1){\frac{1-\alpha}{\alpha}}}_{\mathcal{H},m}\otimes A^{ (1){\frac{1-\alpha}{\alpha}}}_{\mathcal{H},n},\dots,A^{ (r){\frac{1-\alpha}{\alpha}}}_{\mathcal{H},m}\otimes A^{ (r){\frac{1-\alpha}{\alpha}}}_{\mathcal{H},n}))\right|_{\symmetricpower[m+n](\mathcal{H})}  \\
 & \le \geometricmean\left(\left.(A^{ (1){\frac{1-\alpha}{\alpha}}}_{\mathcal{H},m}\otimes A^{ (1){\frac{1-\alpha}{\alpha}}}_{\mathcal{H},n})\right|_{\symmetricpower[m+n](\mathcal{H})},\dots,\left.(A^{(r){\frac{1-\alpha}{\alpha}}}_{\mathcal{H},m}\otimes A^{\alpha (r)}_{\mathcal{H},n})\right|_{\symmetricpower[m+n](\mathcal{H})}\right)  \\
 & \le \geometricmean\left(A^{ (1){\frac{1-\alpha}{\alpha}}}_{\mathcal{H},m+n},\dots,A^{ (r){\frac{1-\alpha}{\alpha}}}_{\mathcal{H},m+n}\right)  \\
 & = A^{\frac{1-\alpha}{\alpha}}_{\mathcal{H},m+n}.
\end{split}
\end{equation}
\Cref{it:observabletensorproduct,it:observabledirectsum} can be proved in a similar way.
\end{proof}
\begin{remark}
    If the observables $A^{ (i)}_{\mathcal{H},n}$ pairwise commute (e.g. if we consider only noncrossing bipartitions), then 
    $A_{\mathcal{H},n}:=\geometricmean(A_{\mathcal{H},n},\dots,A_{\mathcal{H},n})$.
\end{remark}

\begin{example}\label{ex:allbipart}
Recall that $B_k$ denotes the set of bipartitions of $[k]$, i.e. unordered pairs $b=\{S,[k]\setminus S\}$ where $S\subseteq[k]$, $S\neq\emptyset$ and $S\neq[k]$. For a bipartition $b=\{S,[k]\setminus S\}$ we consider the observables
\begin{equation}
A^{(b)}_{\mathcal{H},n}=\left.\sum_{\lambda\vdash n}2^{n\entropy(\lambda/n)}P^{\mathcal{H}_{S}}_\lambda\otimes P^{\mathcal{H}_{[k]\setminus S}}_\lambda\right|_{\symmetricpower[n](\mathcal{H})},
\end{equation}
which satisfy \cref{it:observableisometry,it:observablebound,it:observablesupermultiplicative,it:observabletensorproduct,it:observabledirectsum} by \cref{prop:groupedobservables}.
They give rise to the bipartite exponentiated R\'enyi entropies of entanglement (see \cref{lem:bipartiteasymptotic})
\begin{equation}\label{eq:bipartitionfunctional}
F^{\alpha,A^{(b)}}(\psi)=\Tr(\Tr_S\ketbra{\psi}{\psi})^\alpha.
\end{equation}
By \cref{prop:geometricobservables} we can combine the observables using a $\lvert B_k\rvert=2^{k-1}-1$-variable geometric mean $\geometricmean$ as (by a slight abuse of notation, indexing the arguments of $\geometricmean$ by $B_k$ itself instead of $1,\dots,\lvert B_k\rvert$)
\begin{equation}
A_{\mathcal{H},n}:=\geometricmean\left(\left(\left(A^{(b)}_{\mathcal{H},n}\right)^{\frac{1-\alpha}{\alpha}}\right)_{b\in B_k}\right)^{\frac{\alpha}{1-\alpha}}.
\end{equation}
Consider the upper functional constructed in \cref{thm:upper} with these observables. The resulting logarithmic upper functional \begin{equation}
\begin{split}
E^{\alpha,\geometricmean}(\psi)=\lim_{n\to\infty}\frac{1}{n}\frac{\alpha}{1-\alpha}\log\bra{\psi}^{\otimes n}\geometricmean\left(\left(\left.\sum_{\lambda\vdash n}2^{n\frac{1-\alpha}{\alpha}\entropy(\lambda/n)}P^{\mathcal{H}_{S}}_\lambda\otimes P^{\mathcal{H}_{[k]\setminus S}}_\lambda\right|_{\symmetricpower[n](\mathcal{H})}\right)_{b\in B_k}\right)^{\frac{\alpha}{1-\alpha}} \ket{\psi}^{\otimes n}\\
\\
\end{split}
\end{equation}
interpolates between all the bipartite R\'enyi entropies of entanglement. The operator geometric means parametrize this family of upper functionals, so instead of the observable we show $\geometricmean$ in the index.
\end{example}

\begin{proposition}\label{prop:extension}
The upper functional defined in \eqref{eq:upperdef} is an extension of the upper functional introduced in \cite{vrana2020family}. In particular, for $\alpha\in(0,1)$ and $A_{\mathcal{H},n}=\prod_{j\in [k]} (A^{ (\{\{j\},[k]\setminus\{j\}\})}_{\mathcal{H},n})^{\theta_i}$, we get the logarithmic upper functional \eqref{eq:logupperLOCC}, so the family of functionals in \cref{ex:allbipart} is an extension of this family from elementary bipartitions to all of the bipartitions.
\end{proposition}
\begin{proof}

The following is a generalization of \cite[Eq. (2.13)]{strassen1991degeneration} to more than two terms. For any $x_i\in\mathbb{R}$
\begin{equation}\label{eq:logsup}
    \log\big(\sum_{i\in I} 2^{x_i}\big)=\max_{P\in\distributions(I)}\quad h(P)+ \sum_{i\in I} p_i x_i,
\end{equation}
where $\distributions(I)$ is the set of probability distributions over the set $I$.
By this we have
\begin{equation}
\begin{split}
\log\bra{\psi}^{\otimes n}A^{\frac{1-\alpha}{\alpha}}_{\mathcal{H},n}\ket{\psi}^{\otimes n}
 & = \log\sum_{\lambda\in\partitions[n]^k}2^{n\frac{1-\alpha}{\alpha}H_\theta(\frac{\lambda}{n})+\log\norm{P^\mathcal{H}_\lambda\psi^{\otimes n}}^2}  \\
 & = \sup_{p\in\partitions} \sum_{\lambda\in\partitions[n]^k} p_\lambda\bigg(n\,\frac{1-\alpha}{\alpha}H_\theta\big(\frac{\lambda}{n} \big)+\log\norm{P_\lambda^\mathcal{H}\psi^{\otimes n}}^2 -\log p_\lambda\bigg),
\end{split}
\end{equation}
where the supremum is over the probability distributions over the set $\partitions[n]^k$, so for each $\lambda\in\partitions[n]^k$ we have a probability value $p_\lambda$. For the last term we have the upper bound \newline $h(p_\lambda)=\sum_\lambda -p_\lambda \log p_\lambda \leq \log \lvert\partitions[n]\rvert^k$. The upper bound scales as $\log n$ so $1/n \log \lvert\partitions[n]\rvert^k$ vanishes in the limit, and the upper bound converges to the lower one, which in turn has a simple solution for the optimization problem:
\begin{equation}\label{eq:upperEqProofMiddlestep}
\lim_{n\to\infty}\frac{1}{n}
\log\bra{\psi}^{\otimes n}A^{\frac{1-\alpha}{\alpha}}_{\mathcal{H},n}\ket{\psi}^{\otimes n}
=\lim_{n\to\infty} \max_{\lambda\in\partitions[n]^k}
\bigg(\frac{1-\alpha}{\alpha}H_\theta\big(\frac{\lambda}{n} \big)+\frac{1}{n}\log\norm{P_\lambda^\mathcal{H}\psi^{\otimes n}}^2 \bigg).
\end{equation}

Let $\lambda'^{(n_i)}/n$ be a convergent subsequence of the maximizing $\lambda$ values for given $n$ with the limit $\overline{\lambda}$. The existence of such subsequence follows from the compactness of the finitely supported probability distributions. By \cite[Lemma 3.2]{vrana2020family} we get the lower bound
\begin{equation}
\begin{split}
\lim_{n\to\infty}\frac{1}{n}
\log\bra{\psi}^{\otimes n}A^{\frac{1-\alpha}{\alpha}}_{\mathcal{H},n}\ket{\psi}^{\otimes n}
  \le \sup_{\overline{\lambda}\in\overline{\partitions}^k} \bigg(\frac{1-\alpha}{\alpha}H_\theta\big(\overline{\lambda} \big)-I_\psi(\overline{\lambda}) \bigg).
\end{split}
\end{equation}

Now we show the opposite inequality. There exists a point $\overline{\lambda}'$, for which the evaluation of the right hand side is less than the supremum by at most $\epsilon$. Then by \cite[Lemma 3.3]{vrana2020family} there exists a sequence $\frac{\lambda^{(n)}}{n}\to\overline{\lambda}'$ such that 
\begin{equation}
\begin{split}
& \sup_{\overline{\lambda}\in\overline{\partitions}^k} \bigg(\frac{1-\alpha}{\alpha}H_\theta\big(\overline{\lambda} \big)-I_\psi(\overline{\lambda}) \bigg)+\epsilon \le
\bigg(\frac{1-\alpha}{\alpha}H_\theta\big(\overline{\lambda}' \big)- I_\psi(\overline{\lambda}') \bigg)  \\
=& \lim_{n\to\infty}
\bigg( \frac{1-\alpha}{\alpha}H_\theta\bigg(\frac{\lambda^{(n)}}{n} \bigg)+\frac{1}{n}\log\norm{P^{\mathcal{H}}_{\lambda^{(n)}}\psi^{\otimes n}}^2\bigg).
\end{split}
\end{equation}
The last expression is a lower bound on \eqref{eq:upperEqProofMiddlestep}. This is true for all $\epsilon>0$, then in $\epsilon \to 0$ limit the desired inequality is proven.

\end{proof}

\section{Bounds on the upper functional}\label{sec:lower}

In this section we prove some important bounds and their consequences for the upper functional introduced in \cref{ex:allbipart}.

\begin{proposition}\label{prop:halfsandwichedgmeanbound}
Let $\theta\in\distributions([r])$ convex weights and $\geometricmean$ be a corresponding $r$-variable operator geometric mean. Now let $A^{ (1)}_{\mathcal{H},n},\dots,A^{(r)}_{\mathcal{H},n}$ be families satisfying \cref{it:observableisometry,it:observablebound,it:observablesupermultiplicative,it:observabletensorproduct,it:observabledirectsum} and $A^{\frac{1-\alpha}{\alpha}}_{\mathcal{H},n}=\geometricmean(A^{ (1){\frac{1-\alpha}{\alpha}}}_{\mathcal{H},n},\dots,A^{ (r){\frac{1-\alpha}{\alpha}}}_{\mathcal{H},n})$.
Let $F^{\alpha,A^{(1)}},\dots,F^{\alpha,A^{ (r)}}$ and $F^{\alpha,A}$ be the upper functionals obtained from these families via \cref{thm:upper}. Then for any $\psi$ the inequality
\begin{equation}\label{eq:halfsandwichedgmeanbound}
F^{\alpha,A}(\psi)\le \mathbb{G}(F^{\alpha,A^{ (1)}}(\psi),\dots,F^{\alpha,A^{(r)}}(\psi)) =\prod_{i=1}^r F^{\alpha,A^{ (i)}}(\psi)^{\theta(i)}
\end{equation}
holds.
\end{proposition}
\begin{proof}
Let $E^{\alpha,A}, E^{\alpha,A^{(1)}}$, etc. be the corresponding logarithmic upper functionals. Then for a vector $\psi$ by \cref{rem:gmeanproperties} we have
\begin{equation}
\begin{split}
E^{\alpha,A}(\psi)
 & = \lim_{n\to\infty}\frac{1}{n}\frac{\alpha}{1-\alpha}\log\bra{\psi^{\otimes n}}\geometricmean(A^{ (1){\frac{1-\alpha}{\alpha}}}_{\mathcal{H},n},\dots,A^{ (r){\frac{1-\alpha}{\alpha}}}_{\mathcal{H},n})\ket{\psi^{\otimes n}}  \\
 & \le \lim_{n\to\infty}\frac{1}{n}\frac{\alpha}{1-\alpha}\log\geometricmean(\bra{\psi^{\otimes n}}A^{ (1){\frac{1-\alpha}{\alpha}}}_{\mathcal{H},n}\ket{\psi^{\otimes n}},\dots,\bra{\psi^{\otimes n}}A^{ (r){\frac{1-\alpha}{\alpha}}}_{\mathcal{H},n}\ket{\psi^{\otimes n}})  \\
 & = \lim_{n\to\infty}\frac{1}{n}\frac{\alpha}{1-\alpha}\sum_{i=1}^r\theta(i)\log\bra{\psi^{\otimes n}}A^{(i){\frac{1-\alpha}{\alpha}}}_{\mathcal{H},n}\ket{\psi^{\otimes n}}  \\
 & = \sum_{i=1}^r\theta(i)E^{\alpha,A^{(i)}}(\psi).
\end{split}
\end{equation}
For the upper functionals $F^{\alpha,A}$ this gives the inequality in the statement.
\end{proof}

\begin{example}\label{ex:meanofbipartite}

Now we consider those elements of the asymptotic spectrum, which take into account only one bipartition (see \eqref{eq:bipartitespectrum}). These are lower functionals as well, so we can apply again \cref{prop:functionalbasicconstructions} to get lower functionals interpolating between any bipartition:
\begin{align}\label{eq:interplowerfunc2}
F_{\alpha,\theta}(\psi) & =2^{(1-\alpha)E_{\alpha,\theta}(\psi)} 
\intertext{where}
E_{\alpha,\theta}(\psi) & =\frac{\alpha}{1-\alpha}\log\norm{\psi}^2+\sum_{b}\theta(b)\entropy_{\alpha}(\Tr_b\frac{\vectorstate{\psi}}{\norm{\psi}^2}).  
\end{align}

Note that this lower functional is exactly the right hand side of the inequality in \eqref{eq:halfsandwichedgmeanbound} for $F^{\alpha,A^{(i)}}=F^{\alpha,A^{(b)}}$ being the bipartite observables as in \eqref{eq:bipartitionfunctional}.
\end{example}

\begin{proposition}\label{prop:lowerboundforupperfunc}
 For a $k$-partite system let us consider the same logarithmic upper functional as in \cref{ex:allbipart} and $\alpha\in [1/2,1)$. The following lower bound holds for this functional:

\begin{equation}
    E^{\alpha,\geometricmean}(\psi)  \ge
    \frac{\alpha}{1-\alpha}\log\norm{\psi}^2+\sum_{b}\theta(b)\entropy_{\frac{\alpha}{2\alpha-1}}(\Tr_b\frac{\vectorstate{\psi}}{\norm{\psi}^2}),
\end{equation}
where $\theta\in\distributions([k])$ convex weights correspond to the geometric mean $\geometricmean$.
\end{proposition}

\begin{proof}

First let $\alpha\in [1/2,1)$. By \eqref{eq:geomlowerbyvector} and \cref{lem:bipartiteasymptotic} for a unit vector $\psi$ we get
\begin{equation}
\begin{split}
E^{\alpha,\geometricmean}(\psi)
 & = \lim_{n\to\infty}\frac{1}{n}\frac{\alpha}{1-\alpha}\log\bra{\psi^{\otimes n}}\geometricmean(A^{ (1){\frac{1-\alpha}{\alpha}}}_{\mathcal{H},n},\dots,A^{ (r){\frac{1-\alpha}{\alpha}}}_{\mathcal{H},n})\ket{\psi^{\otimes n}}  \\
 & \ge \lim_{n\to\infty}\frac{1}{n}\frac{\alpha}{1-\alpha}\log\bra{\psi^{\otimes n}}\vectorstate{\psi}^{\otimes n} 
 \ket{\psi^{\otimes n}} 
 2^{-\sum_{i=1}^r\theta(i)\maxrelative{\vectorstate{\psi}^{\otimes n}}{A^{ (i){\frac{1-\alpha}{\alpha}}}_{\mathcal{H},n}}}\\
 &= \lim_{n\to\infty}\frac{1}{n}\frac{-\alpha}{1-\alpha}
 \sum_{i=1}^r\theta(i)\maxrelative{\vectorstate{\psi}^{\otimes n}}{A^{ (i){\frac{1-\alpha}{\alpha}}}_{\mathcal{H},n}}\\
 &= \lim_{n\to\infty}\frac{1}{n}\frac{-\alpha}{1-\alpha}
 \sum_{i=1}^r\theta(i) \log
 \bra{\psi}^{\otimes n}(A^{ (i){\frac{1-\alpha}{\alpha}}}_{\mathcal{H},n})^{-1}\ket{\psi}^{\otimes n}\\
 &= \sum_{i=1}^r\theta(i)
 \entropy_{\alpha'}(\Tr_i\vectorstate{\psi})\\
 &= \sum_{i=1}^r\theta(i)
 \entropy_{\frac{\alpha}{2\alpha-1}}(\Tr_i\vectorstate{\psi})
\end{split}
\end{equation}

In the last part we considered $\alpha'$ with $\frac{1}{\alpha}+\frac{1}{\alpha'}=2$. Then we get that $\alpha'=\frac{\alpha}{2\alpha-1}$ and vica versa, and $\frac{-\alpha}{1-\alpha}=\frac{\alpha'}{1-\alpha'}$. Then \cref{lem:bipartiteasymptotic} can be applied to $\alpha'$.
For an arbitrary vector we get the inequality by the scaling property of the upper functional.
\end{proof}

\begin{proposition}\label{prop:lowerfunctionallowerbound}

For every $\alpha\in(0,1)$, and a fixed $\alpha'>0$ the functional $F_{\alpha,\theta}$ defined as
\begin{align}
F_{\alpha,\theta}(\psi) & = \begin{cases}
2^{(1-\alpha)\lim_{n\to\infty}\frac{1}{n}E_{\alpha,\theta}(\psi^{\otimes n})} & \text{if $\psi\neq 0$}  \\
0 & \text{if $\psi=0$}
\end{cases}  \\
E_{\alpha,\theta}(\psi) & = \sup_{\varphi\le\psi}\left[H_{\alpha',\theta}(\varphi/\norm{\varphi})+\frac{\alpha}{1-\alpha}\log\norm{\varphi}^2\right], \\
H_{\alpha',\theta}(\psi) &= \sum_{i=1}^r\theta(i)\entropy_{\alpha'}(\Tr_i\vectorstate{\psi})
\end{align}

is a lower functional of order $\alpha$. Moreover if $\alpha\in [1/2,1)$ and $\alpha'$ is such that $\frac{1}{\alpha}+\frac{1}{\alpha'}=2$ then $E_{\alpha,\theta}(\psi) \le E^{\alpha,\geometricmean}(\psi)$  holds for the upper functional in \cref{ex:allbipart}.
\end{proposition}
\begin{proof}

The first part of the proof is similar to \cite[Propositions 4.5 and 4.7]{vrana2020family}. Let $\psi_1\in\mathcal{H}$ and $\psi_2\in\mathcal{K}$. If $\varphi_1\le\psi_1$ and $\varphi_2\le\psi_2$, then $\varphi_1\otimes\varphi_2\le\psi_1\otimes\psi_2$, therefore by the additivity of entropy, i.e. $H_{\alpha',\theta}(\psi\otimes\varphi)=H_{\alpha',\theta}(\psi)+H_{\alpha',\theta}(\varphi)$ we have
\begin{equation}
\begin{split}
E_{\alpha,\theta}(\psi_1\otimes\psi_2)
 & \ge H_{\alpha',\theta}\left(\frac{\varphi_1\otimes\varphi_2}{\norm{\varphi_1\otimes\varphi_2}}\right)+\frac{\alpha}{1-\alpha}\log\norm{\varphi_1\otimes\varphi_2}^2  \\
 & = H_{\alpha',\theta}(\varphi_1/\norm{\varphi_1})+\frac{\alpha}{1-\alpha}\log\norm{\varphi_1}^2+
 H_{\alpha',\theta}(\varphi_2/\norm{\varphi_2})+\frac{\alpha}{1-\alpha}\log\norm{\varphi_2}^2.
\end{split}
\end{equation}
Taking the supremum over $\varphi_1$ and $\varphi_2$ we obtain $E_{\alpha,\theta}(\psi_1\otimes\psi_2)\ge E_{\alpha,\theta}(\psi_1)+E_{\alpha,\theta}(\psi_2)$. We also have the linear upper bound $E_{\alpha,\theta}(\psi^{\otimes n})\le n\left[\log\dim\mathcal{H}+\frac{\alpha}{1-\alpha}\log\norm{\psi}^2\right]$, therefore the limit $\lim_{n\to\infty}\frac{1}{n}E_{\alpha,\theta}(\psi^{\otimes n})$ exists and is finite. It also inherits the scaling, superadditivity and monotonicity of $E_{\alpha,\theta}$ (which is monotone by virtue of being defined as a supremum over the lower set generated by its argument). Therefore $F_{\alpha,\theta}$ is monotone and supermultiplicative under the tensor product. Now we use that $H_{\alpha',\theta}(\frac{1}{\sqrt{r}}\unittensor{r})=\log r$ and $H_{\alpha',\theta}(\psi)\le \log\dim\mathcal{H}$ for every $\psi\in\mathcal{H}$:
\begin{equation}
\begin{split}
E_{\alpha,\theta}(\unittensor{r})
 & \ge H_{\alpha',\theta}\left(\frac{\unittensor{r}}{\sqrt{r}}\right)+\frac{\alpha}{1-\alpha}\log\norm{\unittensor{r}}^2  \\
 & = \log r+\frac{\alpha}{1-\alpha}\log r  \\
 & = \frac{1}{1-\alpha}\log r
\end{split}
\end{equation}
and
\begin{equation}
E_{\alpha,\theta}(\unittensor{1})\le \log\dim(\complexes\otimes\dots\otimes\complexes)+\frac{\alpha}{1-\alpha}\log\norm{\unittensor{1}}^2=0,
\end{equation}
therefore $F_{\alpha,\theta}(\unittensor{r})\ge r$ and $F_{\alpha,\theta}(\unittensor{1})=1$. By definition, $F_{\alpha,\theta}(0)=0$.

To prove superadditivity, note that $(\psi\oplus\varphi)^n\ge\unittensor{\binom{n}{m}}\otimes\psi^{\otimes m}\otimes\varphi^{\otimes(n-m)}$ for all $0\le m\le n$, which for $n>0$ implies
\begin{equation}
\begin{split}
\log F_{\alpha,\theta}(\psi\oplus\varphi)
 & = \frac{1}{n}\log F_{\alpha,\theta}((\psi\oplus\varphi)^{\otimes n})  \\
 & \ge \frac{1}{n}\log F_{\alpha,\theta}\left(\unittensor{\binom{n}{m}}\otimes\psi^{\otimes m}\otimes\varphi^{\otimes(n-m)}\right)  \\
 & \ge \frac{1}{n}\log\binom{n}{m}+\frac{m}{n}\log F_{\alpha,\theta}(\psi)+\frac{n-m}{n}\log F_{\alpha,\theta}(\varphi).
\end{split}
\end{equation}
Choosing $m=\lfloor qn\rfloor$ for some $q\in[0,1]$ and letting $n\to\infty$ we obtain the inequality
\begin{equation}
\log F_{\alpha,\theta}(\psi\oplus\varphi)\ge q\log F_{\alpha,\theta}(\psi)+(1-q)\log F_{\alpha,\theta}(\varphi)+h(q),
\end{equation}
and the maximum over $q$ gives $F_{\alpha,\theta}(\psi\oplus\varphi)\ge F_{\alpha,\theta}(\psi)+F_{\alpha,\theta}(\varphi)$.

Now let $\alpha\in [1/2,1)$ and $\alpha'$ be such that $\frac{1}{\alpha}+\frac{1}{\alpha'}=2$, so
we can use \cref{prop:lowerboundforupperfunc}. If $\psi\ge\varphi$ then $F^{\alpha,A}(\psi)\ge F^{\alpha,A}(\varphi)$, therefore
\begin{equation}
\begin{split}
E_{\alpha,\theta}(\psi)
 & = \sup_{\varphi\le\psi}\left[H_{\alpha',\theta}(\varphi/\norm{\varphi})+\frac{\alpha}{1-\alpha}\log\norm{\varphi}^2\right]  \\
 & \le \sup_{\varphi\le\psi}\left[E^{\alpha,A}(\varphi/\norm{\varphi})+\frac{\alpha}{1-\alpha}\log\norm{\varphi}^2\right]  \\
 & = \sup_{\varphi\le\psi}\left[E^{\alpha,A}(\varphi)-\frac{\alpha}{1-\alpha}\log\norm{\varphi}^2+\frac{\alpha}{1-\alpha}\log\norm{\varphi}^2\right]  \\
 & = E^{\alpha,\geometricmean}(\psi).
\end{split}
\end{equation}
\end{proof}

\begin{corollary}\label{cor:specialsubsemiringmonotone}
Let $\alpha\in [1/2,1)$ and $\psi$ be a vector such that $\Tr_S\vectorstate{\psi}$ is proportional to a projection for every subset $S\subseteq[k]$, then $F^{\alpha,\geometricmean}$ in \cref{ex:allbipart} reduces to
\begin{equation}
    F^{\alpha,\geometricmean}(\psi)=\norm{\psi}^{2\alpha}\prod_{b\in B_k}\left(\rank\Tr_S\vectorstate{\psi}\right)^{(1-\alpha)\theta(b)},
\end{equation}
where $\theta\in\distributions(B_k)$ are convex weights corresponding to $\geometricmean$. $F^{\alpha,\geometricmean}$ is multiplicative and additive on the subsemiring generated by such states.
\end{corollary}

\begin{proof}
Let $\alpha\in[1/2, 1)$. If $\psi$ is a vector such that $\Tr_S\vectorstate{\psi}$ is proportional to a projection for every subset $S\subseteq[k]$, then $\entropy_{\alpha',\theta}(\Tr_i\vectorstate{\psi})=\entropy_{\alpha,\theta}(\Tr_i\vectorstate{\psi})$, and both the lower and the upper bounds in \cref{prop:halfsandwichedgmeanbound,prop:lowerfunctionallowerbound} are equal to
\begin{equation}
\sum_{b\in B_k}\theta(b)\log\rank\Tr_S\vectorstate{\psi}+\frac{\alpha}{1-\alpha}\log\norm{\psi}^2.
\end{equation}

The matching lower bound is a lower functional. Note also that by \cref{rem:lowerequalsupper}, the set of tensors on which the upper and the corresponding lower functional are equal is closed under tensor product and direct sums, therefore what we get is an element of the asymptotic spectrum of the subsemiring generated by such states. 

\end{proof}
These states include GHZ states, hypergraph tensors (products of GHZ states on a subset of subsystems \cite{vrana2017entanglement,christandl2016asymptotic}), absolutely maximally entangled states \cite{helwig2012absolute}, and stabilizer states \cite{fattal2004entanglement}.

\begin{corollary}\label{cor:alpha1monotones}
    In the $\alpha\to 1$ limit for a unit vector $\psi$ both the lower and the upper bounds in \cref{prop:halfsandwichedgmeanbound,prop:lowerfunctionallowerbound} are equal to
\begin{equation}
\sum_{b\in B_k}\theta(b)\entropy(\Tr_S\vectorstate{\psi}).
\end{equation}
\end{corollary}

\begin{remark}
    It is not clear in general that two different operator geometric means which determine the same convex weights give different upper functionals. In the case when we take into account only noncrossing bipartitions, the observables commute pairwise, and the geometric mean reduces to the form in \cref{rem:gmeanproperties}. In those cases the weights are enough to parametrize this family of upper functionals.
\end{remark}

\section*{Acknowledgement}

This work was partially funded by the National Research, Development and Innovation Office of Hungary via the research grants K124152, by the \'UNKP-21-5 New National Excellence Program of the Ministry for Innovation and Technology, the J\'anos Bolyai Research Scholarship of the Hungarian Academy of Sciences, and by the Ministry of Innovation and Technology and the National Research, Development and Innovation Office within the Quantum Information National Laboratory of Hungary.

\bibliography{references}{}

\end{document}